\documentclass[10pt,a4paper,twocolumn]{article}

\usepackage[english]{babel}
\usepackage[T1]{fontenc}
\usepackage[utf8]{inputenc}
\usepackage{microtype}
\usepackage[a4paper,margin=2.54cm]{geometry}

\usepackage{amsmath,amssymb,amsthm,mathtools,bm}

\usepackage{booktabs,array,tabularx}
\usepackage{ragged2e}
\usepackage{enumitem}
\usepackage{caption}
\newcolumntype{Y}{>{\RaggedRight\arraybackslash}X}
\newcolumntype{L}[1]{>{\RaggedRight\arraybackslash}p{#1}}

\usepackage{graphicx}
\usepackage{xcolor}
\usepackage{placeins}

\usepackage{xurl}
\usepackage{csquotes}
\usepackage{hyperref}

\usepackage[backend=biber,style=ieee,citestyle=numeric-comp,sorting=none,maxbibnames=99]{biblatex}
\hypersetup{
  hidelinks,
  pdftitle={Rand-SEMI-QAOA: Finite-Budget Depth-One MaxCut Ensembles on Compressed Quantum Registers},
  pdfauthor={Emilio Semre and Steven Frankel},
  pdfsubject={Rand-SEMI-QAOA ensembles for depth-one MaxCut on compressed quantum registers},
  pdfkeywords={
    Rand-SEMI-QAOA,
    SEMI-QAOA,
    QRAO,
    QAOA,
    mutually unbiased bases,
    stabilizer states,
    MaxCut
  }
}

\theoremstyle{plain}
\newtheorem{theorem}{Theorem}[section]
\newtheorem{lemma}[theorem]{Lemma}
\newtheorem{proposition}[theorem]{Proposition}

\theoremstyle{definition}
\newtheorem{definition}[theorem]{Definition}

\theoremstyle{remark}

\newcommand{\R}{\mathbb R}
\newcommand{\F}{\mathbb F}

\newcommand{\E}{\mathbb E}
\newcommand{\Prob}{\mathbb P}

\newcommand{\ii}{\mathrm{i}}

\newcommand{\ket}[1]{\lvert #1\rangle}
\newcommand{\bra}[1]{\langle #1\rvert}

\newcommand{\MUB}{\mathrm{MUB}}

\newcommand{\SEMIQAOA}{\mbox{SEMI-QAOA}}
\newcommand{\RandSEMIQAOA}{\mbox{Rand-SEMI-QAOA}}
\newcommand{\Family}{\mathcal F}
\newcommand{\calH}{\mathcal H}

\DeclareMathOperator{\Tr}{tr}

\DeclareMathOperator{\supp}{supp}
\DeclareMathOperator{\wt}{wt}
\DeclareMathOperator{\rank}{rank}

\title{Rand-SEMI-QAOA: Finite-Budget Depth-One MaxCut Ensembles on Compressed Quantum Registers}
\author{%
Emilio Semre\\\small Department of Computer Science, Technion--Israel Institute of Technology, Haifa, Israel\\\small \href{mailto:abed.semre@campus.technion.ac.il}{abed.semre@campus.technion.ac.il}
\and
Steven Frankel\\\small Faculty of Mechanical Engineering, Technion--Israel Institute of Technology, Haifa, Israel}
\date{}
\begin{document}
\maketitle
\begin{abstract}
We introduce Rand-SEMI-QAOA, a finite-budget QRAO--QAOA ensemble for
MaxCut based on a $(3,1)$-QRAC relaxation.  The method samples labels
uniformly without replacement from the product-$X$ family of an implemented
Galois stabilizer mutually unbiased basis (MUB) system.  Each selected
state--mixer pair is optimized independently for relaxed QRAO energy and
then evaluated by deterministic Pauli-sign decoding.  Exhaustive scans of the
implemented noncomputational MUB catalog place the product-$X$ family first in
18 of 20 validated family-mean cells and in every tested cell for
$r=5,6,7$.  On a complete cohort of $2{,}400$ random connected 3-regular
MaxCut instances with $n\in\{18,20,22\}$, the capped matched-cardinality
schedule attains a graph-mean decoded best-of-set approximation ratio of
$0.9421$ with one QAOA layer, while the $K=r^2$ schedule attains $0.9319$.
These statistics are conditional on the disclosed frozen selector pools and
do not estimate variability over selector seeds.  An exact gauge identity
shows that product-family labels generate the orbit of the QRAO Hamiltonian
under an $r$-dimensional sign-gauge group.  For common angles and relaxed
energy, deterministic syndrome analysis proves branchwise dephasing, while an
anisotropic Gaussian surrogate controls the coherent label-averaged response
on the scale $\beta=b/r$.  The theory does not order independently optimized
decoded maxima.  The results are finite-size and resource-explicit and do not
establish quantum advantage.
\end{abstract}


\section{Introduction}
\label{sec:introduction}

Combinatorial optimization problems (COPs) seek the best solution from a discrete set of possibilities and arise throughout science, engineering, and computer science.  The Quantum Approximate Optimization Algorithm (QAOA) is a leading variational quantum approach to such problems, but at shallow depth its expressivity and achievable solution quality can be limited.  MaxCut, which asks for a partition of a graph that maximizes the number of edges crossing between the two parts, is a standard benchmark for studying these limitations.  For standard QAOA on 3-regular MaxCut, one shared cost angle and one shared mixer angle have a worst-case $p=1$ guarantee of $0.6924$ \cite{FarhiGoldstoneGutmann2014,WurtzLove2021} and give mean expected-energy approximation ratios $0.7617$ and $0.7562$ on the $n=50$ and $100$ triangle-free random 3-regular cohorts studied by Herrman et al.\ \cite{Herrman2022MultiAngle}.  Fixed-parameter concentration, transferable schedules, and optimal-parameter concentration provide additional structure at low depth \cite{BrandaoEtAl2018,WurtzLykov2021,AkshayEtAl2021}.  More broadly, however, depth alone does not specify the resource cost or output statistic of a quantum optimization pipeline: shallow performance can also be supplemented by additional parameters, classical preprocessing, postprocessing, recursion, or an ensemble of quantum ansatzes.

Multi-angle QAOA assigns a separate parameter to each edge and vertex.  At
$p=1$ it reports a mean expected approximation ratio of $0.9257$ over all
connected nonisomorphic eight-vertex graphs, while reporting $0.8123$ and
$0.8000$ on triangle-free random 3-regular graphs at $n=50$ and $100$
\cite{Herrman2022MultiAngle}.  Regularized warm-started QAOA (RWS-QAOA)
constructs an instance-dependent tilted product state by nonconvex classical
optimization; its hardware study uses 100 warm-start initializations and fixed
QAOA parameter schedules, and reports measured-bitstring ratios near $0.97$ at
$p=1$ on five $n=96$ random 3-regular instances, with a single-step local-search
variant reported separately \cite{He2026RWSQAOA}.  Other shallow pipelines
globally round QAOA correlations \cite{DupontSundar2024QRR,Dupont2025QRRStar},
construct and sample a classical surrogate distribution
\cite{WyboFinzgar2026QISS}, or recursively eliminate variables using repeated
quantum calls \cite{Bravyi2020RQAOA,Finzgar2024QIRO,KondoEtAl2025}.  These
methods are relevant comparators, but they allocate different resources and
report noninterchangeable statistics.

\begin{table*}[t]
\centering
\caption{Resource-model context for shallow quantum-assisted MaxCut.  The rows
are not a numerical leaderboard: expected energies, measured bitstrings,
globally rounded correlations, surrogate samples, recursive outputs, and
decoded best-of-set ratios are different objects on different graph ensembles.}
\label{tab:related-resources}
\footnotesize
\setlength{\tabcolsep}{2.5pt}
\par\noindent\makebox[\linewidth][c]{%
\resizebox{0.96\linewidth}{!}{%
\begin{tabular}{L{0.13\textwidth}L{0.18\textwidth}L{0.34\textwidth}L{0.25\textwidth}}
\toprule
Method & Shallow quantum object & Instance-dependent resources beyond $p$ & Reported output \\
\midrule
Standard QAOA \cite{FarhiGoldstoneGutmann2014} & shared $\gamma,\beta$ per layer & angle optimization and direct sampling & expected cut or sampled bitstring \\
Multi-angle QAOA \cite{Herrman2022MultiAngle} & one layer with $|E|+|V|$ angles & high-dimensional angle optimization with multiple starts & expected-energy approximation ratio \\
RWS-QAOA \cite{He2026RWSQAOA} & tilted product state followed by shallow QAOA & regularized nonconvex warm-start optimization; fixed QAOA schedule; optional local search & measured bitstrings \\
QRR/QRR* \cite{DupontSundar2024QRR,Dupont2025QRRStar} & $p=1$ two-point correlations & global correlation-matrix eigensolver and sign rounding; QRR* adds greedy local search & classically rounded cuts \\
QISS \cite{WyboFinzgar2026QISS} & shallow QAOA or RWS correlators & classical factor-model construction and Markov chain Monte Carlo & surrogate-generated cuts \\
Recursive QAOA/QRAO \cite{Bravyi2020RQAOA,KondoEtAl2025} & repeated shallow calls & correlation-driven variable elimination and a terminal solve & reconstructed final cut \\
\shortstack[l]{Rand-SEMI-\\{}QAOA\\{}(this work)} & one layer per candidate on $r$ relaxed qubits & finite product-label budget, one angle optimization per candidate, QRAC decoding, and best-of-set selection & graph-mean decoded best-of-set ratio \\
\bottomrule
\end{tabular}%
}}
\par
\end{table*}

Quantum random access optimization (QRAO) provides a different shallow ansatz
space.  It maps several classical variables to noncommuting Pauli observables
on fewer qubits, producing a compressed noncommuting relaxation of the original
objective \cite{FullerEtAl2024,TeramotoEtAl2023}.  Compression removes the
canonical computational-basis encoding and makes the initial state and mixer
part of the ansatz design.  This is consistent with the alternating-operator
view of QAOA and with results showing that state--mixer alignment can improve
shallow constrained optimization \cite{HadfieldEtAl2019,HeEtAl2023Alignment}.
Prior alternating-operator work for QRAO compares matched $X$, $Y$, and $Z$
product state--mixer pairs and transferable fixed parameters on random
3-regular MaxCut \cite{HeEtAl2025QRAO}.

A complete stabilizer MUB system supplies a finite, algebraically controlled
catalog of aligned states and mixers.  In dimension $d=2^r$, a complete system
contains $d+1$ orthonormal bases with squared cross-basis overlap $1/d$
\cite{WoottersFields1989,BandyopadhyayEtAl2002}; for qubits it can be
constructed from a symplectic spread of maximal commuting Pauli classes
\cite{LawrenceEtAl2002,Abdukhalikov2015}.  Excluding the computational basis
leaves $d$ families and $d^2=4^r$ state--mixer pairs.  MUBs have also been used
as discrete probes of variational landscapes \cite{MeiromAlfassiMor2024}; our
companion work studies basis-union coverage and adaptive family search in a
distinct numerical setting \cite{SemreFrankel2026}.

The full $4^r$ catalog is used here only for discovery.  It identifies the
unique product-$X$ family of this construction as the stable leading family on
the validated $r=5,6,7$ registers.  We therefore define Stabilizer-Encoded
MUB-Informed QAOA (\SEMIQAOA) by restricting the candidate state--mixer pairs to
\(
\{(Z^v\ket{+}^{\otimes r},B_{0,v}):v\in\F_2^r\}
\).
\RandSEMIQAOA{} chooses a finite subset of labels uniformly without replacement
under a declared schedule $K(r)$.  The random label set is fixed by register
size and selector seed, not by an instance-specific approximate cut; every
selected candidate nevertheless receives its own classical angle optimization.

On the complete $n=18,20,22$ cohort of 2,400 random connected 3-regular graphs,
matched-budget \RandSEMIQAOA{} attains mean decoded ratio $0.9421$ at $p=1$,
and the $K(r)=r^2$ schedule attains $0.9319$.  The structured cross-family MUB
selector $\Xi_{123}$ attains $0.9439$, and the single
$\ket{+}^{\otimes r}$ QRAO baseline attains $0.7108$.  The first three numbers
are best-of-set statistics over independently optimized candidates.  They are
not single-circuit approximation ratios, worst-case guarantees, or evidence of
a total-resource advantage.

Exhaustive family scans, a zero-angle control, and a one-layer surrogate analysis support the product-family restriction.  First, exhaustive enumeration of
all noncomputational MUB families places the product family first in 18 of 20
validated $(r,p)$ family-mean cells and at every tested depth for $r=5,6,7$.
Second, a zero-angle control gives mean product-family ratio $0.4090$ before
QAOA and $0.7623$ after one optimized layer; at $K=64$, the exact expected
best-of-$K$ ratio is $0.602$ before QAOA and $0.923$ after it.  Third, the theory
identifies a one-layer relaxed-energy mechanism: typical entangled families
have extensive mixer syndromes whose individual branches dephase, and an exact
anisotropic Gaussian polymer expansion controls the coherent branch sum on the
natural entangled-family mixer scale.  The theory does not prove decoded-ratio
ordering or any $p>1$ result.

The main contributions are:
\begin{enumerate}[leftmargin=*,itemsep=2pt]
  \item the \SEMIQAOA{} framework and its uniform without-replacement
  \RandSEMIQAOA{} selector, with candidate and optimizer resources stated
  explicitly;
  \item complete-cohort depth-one decoded ratios $0.9421$ and $0.9319$ under
  matched and $r^2$ random product-label budgets, respectively, together with
  an explicitly censored scale extension through completed $n=34$ scans;
  \item an exhaustive stabilizer-MUB discovery campaign, a no-QAOA dynamical
  control, exact finite-population best-of-$K$ curves, and completeness-gated
  energy--decoder diagnostics; and
  \item deterministic branchwise syndrome dephasing and an annealed Gaussian
  separation theorem for the one-layer relaxed-energy response.
\end{enumerate}

We do not claim quantum advantage, lower total computation, or superiority to
Goemans--Williamson or modern classical MaxCut solvers
\cite{GoemansWilliamson1995}.  The result is instead a resource-explicit route
to high decoded performance from compressed, depth-one, product-family
QRAO--QAOA ensembles.

\section{The \texorpdfstring{\SEMIQAOA}{SEMI-QAOA} framework}
\label{sec:framework}

\subsection{QRAC relaxation of a spin objective}

Write a binary optimization objective in spin variables $z_i\in\{\pm1\}$ as
\begin{equation}
  C(z)=c+\sum_i a_i z_i+\sum_{i<j}b_{ij}z_i z_j.
  \label{eq:spin-objective}
\end{equation}
For a $(k,1)$ quantum random access code with fixed $k\le3$, variable $i$ is
assigned a one-qubit Pauli $P_i\in\{X_q,Y_q,Z_q\}$, with at most $k$
variables on each relaxed qubit and adjacent graph vertices assigned to
different qubits.
The relaxed Hamiltonian is
\begin{equation}
  H=cI+\sqrt{k}\sum_i a_iP_i+k\sum_{i<j}b_{ij}P_iP_j.
  \label{eq:qrao-hamiltonian}
\end{equation}
For MaxCut, $a_i=0$ and every graph edge contributes a weight-two Pauli term.
The implementation uses an exact minimum coloring and splits every color
class into chunks of at most three variables; the axes in a chunk are distinct
and cyclically ordered.  This guarantees that no objective edge collapses
onto one relaxed qubit.  QRAO, its magic-state approximation guarantees, and
deterministic Pauli rounding were introduced by Fuller et al.
\cite{FullerEtAl2024}; Teramoto et al.\ extend the theory to additional QRAC
relaxations \cite{TeramotoEtAl2023}.  Subsequent work includes recursive QRAO
\cite{KondoEtAl2025} and explicit noise and shot-complexity analysis
\cite{TamuraEtAl2024}.  The connection to quantum random access codes
originates in dense quantum coding \cite{AmbainisEtAl2002}.

The reported decoded score is the deterministic Pauli-sign rounding of
Ref.~\cite{FullerEtAl2024} used by the implementation.  For MaxCut graph $G$,
with exact optimum
$C_{\max}(G)$, the approximation ratio is
\begin{equation}
  \alpha(\rho;G)=\frac{C(\widehat z(\rho))}{C_{\max}(G)},
  \qquad
  \widehat z_i=\operatorname{sgn}_0\!\bigl(\Tr(\rho P_i)\bigr),
  \label{eq:decoded-ratio}
\end{equation}

Throughout this work the deterministic tie convention is
\begin{equation}
\operatorname{sgn}_0(x)=
\begin{cases}
+1, & x\ge 0,\\
-1, & x<0.
\end{cases}
\label{eq:decoder-tie}
\end{equation}
Thus an exactly vanishing one-body expectation is decoded as $+1$.  The
convention is fixed before candidate evaluation.  Because the zero-angle
control contains exact zero expectations, that control is used only as a
convention-dependent diagnostic and not as primary evidence for the
post-optimization ordering.

For a pure state, we write $\alpha(\Psi;G)$ as shorthand for
$\alpha(\ket{\Psi}\!\bra{\Psi};G)$.  This nonlinear map is distinct
from the relaxed energy $\Tr(\rho H)$.  In particular, maximizing the latter
need not maximize Eq.~\eqref{eq:decoded-ratio}.

\subsection{Stabilizer MUB catalog}

Let $V=\F_2^r$ and $d=2^r$.  The $r$-qubit nonidentity Paulis can be
partitioned into $d+1$ maximal commuting classes whose joint eigenbases are
mutually unbiased \cite{LawrenceEtAl2002}.  A binary symplectic spread gives symmetric
matrices $M_j\in\F_2^{r\times r}$ indexed by $j\in V$, such that
$M_j-M_{j'}$ is nonsingular whenever $j\ne j'$.  In the Hermitian Pauli
convention $P(x,z)$, define the maximal commuting subgroup
\begin{equation}
  L_j=\{P(x,M_jx):x\in V\}.
\end{equation}
Its joint eigenstates form one basis
$\Family_j=\{\ket{\psi_{j,v}}:v\in V\}$, and the $d$ families together with
the computational basis form a complete MUB system
\cite{WoottersFields1989,BandyopadhyayEtAl2002,Abdukhalikov2015}.  We enumerate only the
$d$ noncomputational families, hence $d^2=4^r$ states per graph.

Writing the symmetric matrix entries as $(M_j)_{qq'}$, one compatible
Clifford preparation is
\begin{align}
  C_{j,v}={}&
  \left(\prod_{1\le q<q'\le r}CZ_{qq'}^{(M_j)_{qq'}}\right)\nonumber\\[-2pt]
  &{}\times
  \left(\prod_{q=1}^{r}S_q^{(M_j)_{qq}}\right)
  Z^vH^{\otimes r}.
  \label{eq:mub-clifford}
\end{align}
The corresponding state is
$\ket{\psi_{j,v}}=C_{j,v}\ket{0}^{\otimes r}$, up to a global phase.  At
field element $j=0$, all phase and controlled-$Z$ exponents vanish, so
\begin{equation}
  \ket{\psi_{0,v}}=Z^v\ket{+}^{\otimes r}.
  \label{eq:product-family}
\end{equation}
No $S$ gates and no $CZ$ gates occur in this product family.  Its only
label-dependent state-preparation operation is the $Z^v$ sign pattern applied
to $\ket{+}^{\otimes r}$.
Every stabilizer generator is therefore one-local.  In the implemented
spread catalog, direct inspection confirms that no other noncomputational
family has exclusively one-local stabilizer generators.  We use ``family 1'' for the one-based implementation label and
``$j=0$'' for its field label.

Let $g_{j,v,1},\ldots,g_{j,v,r}$ be independent commuting signed Hermitian
Paulis satisfying $g_{j,v,q}\ket{\psi_{j,v}}=\ket{\psi_{j,v}}$.  The matched
mixer convention used in the experiments is
\begin{equation}
  B_{j,v}=-\sum_{q=1}^r g_{j,v,q},
  \qquad B_{j,v}\ket{\psi_{j,v}}=-r\ket{\psi_{j,v}}.
  \label{eq:mixer-convention}
\end{equation}
For $j=0$, $g_{0,v,q}=(-1)^{v_q}X_q$, hence the matched signed-$X$ mixer is
\begin{equation}
  B_{0,v}=-\sum_{q=1}^{r}(-1)^{v_q}X_q.
  \label{eq:product-mixer}
\end{equation}
The product-family ansatz catalog is therefore
\begin{equation}
 \mathcal S_r=
 \left\{
   \left(Z^v\ket{+}^{\otimes r},B_{0,v}\right):v\in\F_2^r
 \right\}.
 \label{eq:semi-family}
\end{equation}
The depth-$p$ state is
\begin{equation}
\begin{split}
  \ket{\Psi_{j,v}(\bm\gamma,\bm\beta)}
  ={}&e^{-\ii\beta_p B_{j,v}}e^{-\ii\gamma_p H}\cdots\\
     &{}\times e^{-\ii\beta_1 B_{j,v}}e^{-\ii\gamma_1 H}
  \ket{\psi_{j,v}}.
\end{split}
  \label{eq:qaoa-state}
\end{equation}
\begin{definition}[\SEMIQAOA]
Given a QRAO Hamiltonian on $r$ relaxed qubits, \SEMIQAOA{} uses candidate
state--mixer pairs only from the product-$X$ stabilizer-MUB family
$\mathcal S_r$ in Eq.~\eqref{eq:semi-family}.  A label-selection rule chooses
a finite set $S_r\subseteq\F_2^r$; every selected pair receives an independent
depth-$p$ angle optimization, is decoded by the stated classical rule, and
competes in a best-of-set selection.  The product-family restriction is part
of the definition.  Here ``Stabilizer-Encoded'' refers to the stabilizer-state
MUB ansatz representation, not to a quantum error-correcting code.
\end{definition}

\begin{definition}[\RandSEMIQAOA]
For a declared candidate-budget schedule $K(r)$, \RandSEMIQAOA{} is
\SEMIQAOA{} with $S_r$ sampled uniformly without replacement from $\F_2^r$,
with
\begin{equation}
 |S_r|=\min\{K(r),2^r\}.
 \label{eq:rand-semi-budget}
\end{equation}
When the requested budget reaches $2^r$, the selector is the complete product
family and the random draw is trivial.  The name denotes the random
product-label rule independently of the particular schedule $K(r)$.
\end{definition}
Here $p$ counts alternating cost--mixer layers for one candidate.  In the
performance experiment, each random product-label set is fixed for its
register size and chosen without using an instance-specific approximate cut;
the angles of every label are optimized independently on each graph.  State
preparation, the number of candidates, classical angle optimization, and
decoded-score selection are therefore resources in addition to $p$.
Equation~\eqref{eq:mixer-convention} fixes an otherwise easy-to-miss sign:
for a Pauli $T$ anticommuting with a generator $g$,
\begin{equation}
 e^{-\ii\beta g}T e^{\ii\beta g}
 =\cos(2\beta)T-\ii\sin(2\beta)gT.
 \label{eq:pauli-conjugation}
\end{equation}
All branch formulas below follow this physical convention.  Using
$B=+\sum g$ would complex-conjugate the radial phases and reverse the small
cost-angle slope.

For $v=0$ in the product family,
$\ket{\psi_{0,0}}=\ket{+}^{\otimes r}$ and
$B_{0,0}=-\sum_qX_q$.  The baseline in every campaign is exactly this one
QRAO-QAOA label; it is not standard QAOA on the original $n$-qubit diagonal
MaxCut Hamiltonian.

The candidate catalog in this work consists of the noncomputational families
of the implemented Galois stabilizer-MUB construction. The computational-basis
product family is excluded by design because its states do not have uniform
computational-basis amplitude support, which is imposed here as an
initialization requirement. Accordingly, the exhaustive family comparison
does not claim that a matched-\(Z\) computational-basis ansatz is inferior; it
is an untested control outside the present catalog.

\subsection{Exact product-label gauge equivalence}
\label{subsec:product-label-gauge}

The product-family label has an exact gauge interpretation. Let
\begin{equation}
\begin{aligned}
 W_v&=Z_1^{v_1}\cdots Z_r^{v_r},\\
 B_0&=-\sum_{q=1}^{r}X_q,
 & H_v&=W_v H W_v .
\end{aligned}
\end{equation}
Then $\ket{\psi_{0,v}}=W_v\ket{+}^{\otimes r}$ and
$B_{0,v}=W_vB_0W_v$.

\begin{proposition}[Product-label gauge equivalence]
\label{prop:product-label-gauge}
For every depth \(p\) and all angle vectors
\(\boldsymbol{\gamma},\boldsymbol{\beta}\),
\begin{equation}
\lvert\Psi_{0,v}(\boldsymbol{\gamma},\boldsymbol{\beta};H)\rangle
=
W_v\,
\lvert\Psi_{0,0}(\boldsymbol{\gamma},\boldsymbol{\beta};H_v)\rangle .
\label{eq:product-label-gauge}
\end{equation}
For a logical variable \(i\) encoded by \(P_i\in\{X_{q(i)},Y_{q(i)},Z_{q(i)}\}\),
define
\begin{equation}
s_i(v)=
\begin{cases}
(-1)^{v_{q(i)}}, & P_i\in\{X_{q(i)},Y_{q(i)}\},\\
1, & P_i=Z_{q(i)}.
\end{cases}
\end{equation}
The corresponding one-body expectation satisfies
\begin{equation}
\langle P_i\rangle_{0,v;H}
=
s_i(v)\,\langle P_i\rangle_{0,0;H_v}.
\label{eq:gauge-one-body}
\end{equation}
Whenever the expectation is nonzero, the decoded spin transforms by the same
factor $s_i(v)$.  At exact zero, Eq.~\eqref{eq:decoder-tie} is applied to each
expectation separately, and no multiplicative decoded-spin relation is
asserted.
\end{proposition}

\begin{proof}
Conjugation by \(W_v\) gives
\(B_{0,v}=W_vB_0W_v\) and
\(\exp(-\mathrm{i}\gamma H)W_v
 =W_v\exp(-\mathrm{i}\gamma H_v)\).
Applying these identities layer by layer moves one \(W_v\) to the left of the
entire circuit and yields Eq.~\eqref{eq:product-label-gauge}. Conjugation by
\(Z_{q(i)}\) changes the sign of \(X_{q(i)}\) and \(Y_{q(i)}\) and leaves
\(Z_{q(i)}\) invariant, which proves Eq.~\eqref{eq:gauge-one-body}.
\end{proof}

Thus Rand-SEMI-QAOA may equivalently be viewed as searching the orbit of $H$
under the $r$-dimensional sign-gauge group while using the fixed state
$\ket{+}^{\otimes r}$ and the fixed $X$ mixer.
This identity does not imply an ordering of labels or of decoded optima.

\subsection{Four distinct performance objects}
\label{sec:metrics}

We keep four objects separate throughout:

\begin{definition}[Family mean]
For graph $G$ and fixed depth, the decoded family mean is
\begin{equation}
  \overline\alpha_j(G)=d^{-1}\sum_{v\in V}
  \alpha\!\left(\Psi_{j,v}^{\star};G\right),
  \label{eq:family-mean}
\end{equation}
where the star denotes independent angle optimization for that label.  The
graph is the inferential unit; reported family means average
$\overline\alpha_j(G)$ over graphs.
\end{definition}

\begin{definition}[Best of a candidate set]
For an offered set $S$ of state labels,
\begin{equation}
  \alpha_{\max}(S;G)=\max_{(j,v)\in S}
  \alpha\!\left(\Psi_{j,v}^{\star};G\right).
  \label{eq:best-of-set}
\end{equation}
This is an order statistic and grows mechanically with the number and
dependence structure of candidates.
\end{definition}

The \emph{relaxed energy} is
$E_{j,v}(\bm\gamma,\bm\beta)=\langle\Psi_{j,v}|H|\Psi_{j,v}\rangle$.
The theoretical \emph{family response} is a normalized, label-averaged change
in this energy.  Neither quantity equals the decoded ratio.  Finally, the
\emph{branch amplitude} introduced in Sec.~\ref{sec:mechanism} is an exact
trace term in the energy response; bounding one branch does not bound a
coherent sum unless its phases are controlled.

\subsection{Candidate cost and exact \texorpdfstring{best-of-$K$}{best-of-K} control}

The exhaustive noncomputational catalog contains $4^r$ labels.  A candidate
strategy pays one independent angle optimization per distinct label in its
union.  The product family alone contains $2^r$ states, and the sampled
selectors of Sec.~\ref{sec:shallow-performance} use at most a few hundred
labels over the studied range.  We use ``best-of-$K$'' for a maximum over one
offered set of $K$ candidates; it does not mean $K$ samples from one optimized
circuit.

For a finite family with decoded scores
$y_{(1)}\le\cdots\le y_{(N)}$, the expected maximum of a uniform sample of
$K$ labels without replacement is available exactly:
\begin{equation}
 \E[\max_K]
 =\sum_{i=K}^{N} y_{(i)}
   \frac{\binom{i-1}{K-1}}{\binom{N}{K}}.
 \label{eq:best-k-exact}
\end{equation}
We apply Eq.~\eqref{eq:best-k-exact} within each graph and only then average
over graphs.  It gives a resource curve without Monte Carlo noise and exposes
how much apparent improvement is due solely to testing more states.

\section{Finite-budget \texorpdfstring{\RandSEMIQAOA}{Rand-SEMI-QAOA}}
\label{sec:shallow-performance}
\label{sec:sampling}

The main performance experiment asks how much decoded solution quality can be
obtained from one cost--mixer alternation per candidate, without first
computing a problem-specific approximate cut to initialize the quantum state.
It compares two \RandSEMIQAOA{} budget schedules with a cross-family MUB
selector and a single-label baseline.  Each arm is a hybrid best-of-$K$
procedure, not a single circuit run: every candidate receives an independent
angle optimization, and the candidate with the largest stored decoded ratio
is retained.  Thus $p=1$ describes the quantum alternation depth of each
candidate, while $K$, the angle-search budget, state preparation, and decoding
are separate resources.

\subsection{Rand-SEMI-QAOA budgets and comparison selectors}

The exhaustive study in Sec.~\ref{sec:exhaustive} motivates three finite
selectors.  The structured MUB selector $\Xi_{123}(r)$ contains a dominant
product-family branch, a smaller basis-2 branch, and four basis-3 labels only
for $r\le6$.  It is not \SEMIQAOA{} because it leaves the product-$X$ family.
With $q=2^{r-2}$, $c=(r-1)\bmod q$, and cyclic distance
$\operatorname{dist}_q$, its two extended branches are
\begin{align}
 &(1,v):\operatorname{dist}_q(v\bmod q,c)\le\rho_1(r),\nonumber\\
 &(2,v):\operatorname{dist}_q(v\bmod q,c)\le
   \left\lfloor\sqrt{3r^2/4}\right\rfloor,
 \label{eq:xi-rule}
\end{align}
where
\begin{equation}
 \rho_1(r)=
 \begin{cases}
 r,&r=4,5,\\
 8,&r=6,\\
 \lfloor19r/8\rfloor,&r\ge7.
 \end{cases}
\end{equation}
The notation records the historical three-branch rule; it does not denote
three complete MUB families.

The numerical constants in $\Xi_{123}$ reproduce a historical cross-family
rule from the preceding discovery campaign.  The rule was frozen before the
complete $2{,}400$-graph scale cohort was evaluated and was not fitted on that
cohort.

Both product-only arms are instances of \RandSEMIQAOA{} under different
budget schedules:
\begin{equation}
\begin{aligned}
 K_{r^2}(r)&=\min\{r^2,2^r\},\\
 K_{\rm match}(r)&=\min\{|\Xi_{123}(r)|,2^r\}.
\end{aligned}
\label{eq:rand-budget-schedules}
\end{equation}
The first is the $r^2$-budget arm; the second is the matched-budget arm and is
the flagship schedule.
For each register size and budget schedule, one random product-label subset is
generated once and reused across graphs and depths; the two budget schedules
use independently generated frozen subsets.  The experiment therefore tests
one frozen pool per schedule rather than averaging over selector draws.
At $r=6,7$, $K_{\rm match}=2^r$: the matched-budget arm contains the complete
product family, the draw is trivial, and its cardinality is smaller than
$|\Xi_{123}|$.  Exact cardinality matching with $\Xi_{123}$ begins at $r=8$.
Consequently the matched Rand-SEMI schedule and $\Xi_{123}$ have $O(r)$
cardinality, whereas the second Rand-SEMI schedule is $O(r^2)$; all are
polynomial substitutes for the $4^r$ discovery catalog.

\begin{definition}[Experimental \RandSEMIQAOA{} protocol]
\label{def:rand-semi}
For a graph $G$, relaxed register size $r$, and declared schedule $K(r)$:
\begin{enumerate}[leftmargin=*,itemsep=1pt]
  \item construct its $(3,1)$-QRAC/QRAO Hamiltonian;
  \item use the fixed set $V_r^{\rm Rand}(K)\subseteq\F_2^r$, sampled
  uniformly without replacement with size $\min\{K(r),2^r\}$;
  \item for every $v\in V_r^{\rm Rand}(K)$, prepare
  $Z^v\ket{+}^{\otimes r}$ and use
  $B_{0,v}=-\sum_{q=1}^r(-1)^{v_q}X_q$;
  \item independently optimize the depth-$p$ QAOA angles for that candidate;
  \item apply the deterministic Pauli-sign decoder and return the candidate
  with the largest decoded MaxCut score.
\end{enumerate}
When $K(r)\ge2^r$, step 2 returns the full product family.  Otherwise the
fixed subset depends on $r$, the budget schedule, and the disclosed selector
seed, not on a classically computed approximate solution of the individual
graph.
\end{definition}

\begin{table*}[t]
 \centering
 \caption{Candidate resources in the scale campaign.  The two Rand-SEMI
 columns use the same uniform product-family mechanism under the $r^2$ and
 matched candidate-budget schedules.  $\Xi_{123}$ is a cross-family structured
 MUB comparator.  Arm overlap means the
 physical union is smaller than the sum of displayed arm sizes.  One
 $\ket{+}^{\otimes r}$ baseline optimization is added per graph.  The executed
 union also contains a legacy arm omitted from the primary comparison; its
 contribution remains in the physical count and is itemized in
 Appendix~\ref{app:experiments}.}
 \label{tab:resource-counts}
 \footnotesize
\begin{tabular*}{\textwidth}{@{\extracolsep{\fill}}ccccccc@{}}
\toprule
$r$ & \shortstack{structured\\$\Xi_{123}$} & \shortstack{$r^2$-budget\\Rand-SEMI} & \shortstack{matched-budget\\Rand-SEMI} & \shortstack{product\\family} & \shortstack{executed\\union$^\dagger$} & \shortstack{noncomputational\\MUB catalog} \\
\midrule
6 & 112 & 36 & 64 & 64 & 112 & 4,096 \\
7 & 180 & 49 & 128 & 128 & 184 & 16,384 \\
8 & 208 & 64 & 208 & 256 & 302 & 65,536 \\
9 & 232 & 81 & 232 & 512 & 414 & 262,144 \\
10 & 256 & 100 & 256 & 1,024 & 528 & 1,048,576 \\
11 & 288 & 121 & 288 & 2,048 & 641 & 4,194,304 \\
12 & 312 & 144 & 312 & 4,096 & 742 & 16,777,216 \\
\bottomrule
\end{tabular*}
\par\smallskip
\begin{minipage}{0.94\textwidth}\footnotesize
$^\dagger$The executed union is the deduplicated physical campaign union of the displayed arms and the legacy $R_E$ arm; one baseline run is additional. Detailed overlap accounting appears in Appendix~\ref{app:experiments}, Table~\ref{tab:resource-legacy}.
\end{minipage}

\end{table*}

\begin{figure*}[t]
 \centering
 \includegraphics[width=\textwidth]{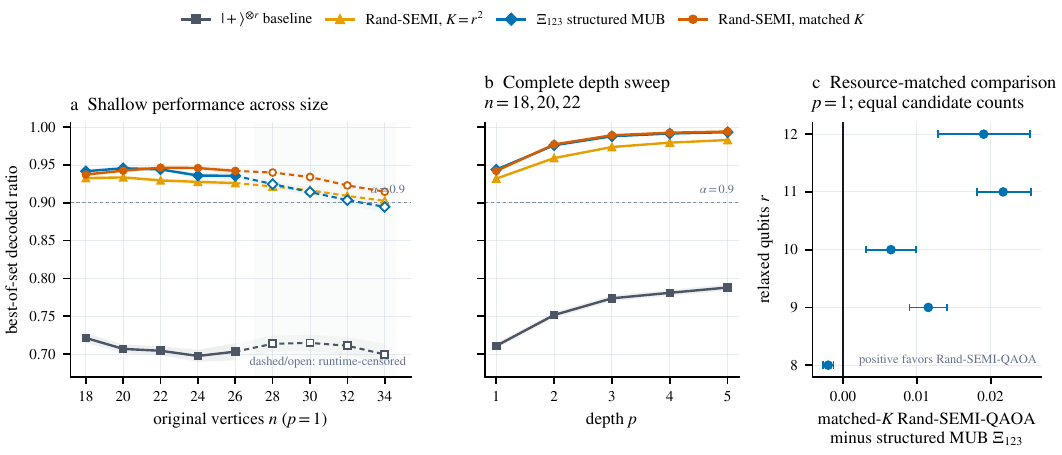}
 \caption{Shallow decoded MaxCut performance under an explicit candidate
 budget.  The legend compares the $r^2$-budget and matched-budget
 \RandSEMIQAOA{} schedules with the structured MUB selector $\Xi_{123}$ and
 the single-state baseline.  \textbf{a}, depth-one best-of-set ratio versus original graph size.
 Solid lines and filled markers denote the complete or compositionally
 reliable portion through $n=26$; dashed lines, open markers, and light
 shading mark the runtime-censored $n\ge28$ analysis freeze.  \textbf{b}, depth sweep
 on the complete $n=18,20,22$ cohorts (2,400 graphs per point).
 \textbf{c}, graph-paired matched-$K$ \RandSEMIQAOA{} minus the structured
 MUB selector $\Xi_{123}$ at $p=1$ where candidate counts
 are exactly matched ($r=8,\ldots,12$).  Whiskers and ribbons are 95\%
 graph-bootstrap intervals.  The $\alpha=0.9$ lines are visual references,
 not theoretical thresholds.  A single legend applies to panels a and b.}
 \label{fig:shallow-performance}
\end{figure*}

For the complete $p=1$ cohort below, $r=6,7,8,9$.  Across those registers,
$K=64$--232 for matched-budget \RandSEMIQAOA, $K=36$--81 for
$r^2$-budget \RandSEMIQAOA, and $K=112$--232 for $\Xi_{123}$.  Because arms overlap, the executed
campaign optimizes 112, 184, 302, or 414 distinct candidate states at these
register sizes, plus one baseline.  The mean physical campaign count is 275.27
candidate state--mixer optimizations per graph, including the baseline.  The complete resource table also
shows the exponentially larger $2^r$ product family and $4^r$
noncomputational MUB catalog.  Candidate search therefore moves work outside
QAOA depth; no total-resource or time-to-solution advantage is claimed.

\subsection{Depth-one performance, scale, and depth persistence}

The scale data combine three campaigns with identical graph generation,
selector, and optimizer settings.  The complete inferential cohort contains
800 random connected 3-regular graphs at each of $n=18,20,22$, hence 2,400
graphs per depth.  Their QRAC packings use $r=6,\ldots,9$ relaxed qubits:
the graph-weighted means are $\E[r/n]=0.3876$ (range $1/3$ to $4/9$)
and $\E[n/r]=2.594$.  This is qubit compression of the
noncommuting Hamiltonian, not a claim about total circuit or optimization
cost.

At $p=1$, let $\alpha_{\rm match}$ and $\alpha_{r^2}$ denote the two
Rand-SEMI-QAOA schedules, $\alpha_{\Xi}$ the structured MUB selector, and
$\alpha_{+}$ the single-state baseline.  Their graph-mean decoded ratios are
\begin{equation}
\begin{aligned}
 \alpha_{\rm match}&=0.9421\ [0.9406,0.9435],\\
 \alpha_{r^2}&=0.9319\ [0.9304,0.9334],\\
 \alpha_{\Xi}&=0.9439\ [0.9424,0.9453],\\
 \alpha_{+}&=0.7108\ [0.7072,0.7143].
\end{aligned}
\label{eq:headline-performance}
\end{equation}
Intervals are 95\% graph-bootstrap intervals.  Every candidate in the first
three rows receives the same $33\times65$ angle grid and three L-BFGS-B local
starts before best-of-$K$ selection on decoded score.  No instance-specific
classically optimized approximate cut or initial state is used as the warm
start.

Figure~\ref{fig:shallow-performance}a separates the complete-cohort statement
from a descriptive scale observation.  Both observed \RandSEMIQAOA{} budget
schedules are above $0.9$ at every displayed $n$; at the 65 completed $n=34$
scans, the matched and $r^2$ schedules are $0.9147$ and $0.9026$.  The
structured MUB selector $\Xi_{123}$
falls from $0.9036$ at $n=32$ to $0.8946$ at $n=34$, so the through-$n=34$
observation applies to the two pure-product selectors, not to every selector.
Moreover, the $n\ge28$ tail is runtime-censored: shorter-register tasks finish
first, and the completed graphs do not reproduce the intended register-size
mixture.  Those points are therefore not unbiased population estimates.

The complete depth sweep in Fig.~\ref{fig:shallow-performance}b remains a
best-of-set experiment.  Matched-budget \RandSEMIQAOA{} rises from $0.9421$
at $p=1$ to $0.9942$ at $p=5$, while the $r^2$-budget instance rises from
$0.9319$ to $0.9829$ and $\Xi_{123}$ from $0.9439$ to $0.9934$.
These $p>1$ results are experimental;
the mechanism theory in Sec.~\ref{sec:mechanism} is one-layer only.

Panel c tests selector geometry at equal candidate count.  At $r=8$,
matched-budget \RandSEMIQAOA{} trails $\Xi_{123}$ by $0.0020$.  At
$r=9,10,11,12$ it leads
by $0.0115$, $0.0065$, $0.0217$, and $0.0190$, respectively, with paired
bootstrap intervals above zero.  The fitted index bands in $\Xi_{123}$ thus
do not provide a persistent advantage over the one fixed random product
subset.  Since only one subset seed is available at each $r$, this is
conditional on those subsets and does not estimate random-selector
variability.

  Over the complete $n=18,20,22$ cohorts,
matched-budget \RandSEMIQAOA{} attains mean decoded ratio $0.9421$ with a
finite ensemble of independently optimized depth-one circuits on compressed
QRAO registers, without an instance-specific classically optimized warm-start
solution.  It is not a single-circuit result, a worst-case guarantee, or a
comparison with classical MaxCut solvers.

\section{Exhaustive MUB discovery}
\label{sec:exhaustive}

\RandSEMIQAOA{} was not motivated by assuming that product states must be
favorable.  Its product-family restriction was distilled from a completed
exhaustive enumeration campaign over the full noncomputational stabilizer-MUB catalog.
The exhaustive all-family search is the discovery procedure, not itself a
\SEMIQAOA{} execution.  This section reports that discovery statistic, which
is a family mean and must not be confused with the best-of-$K$ performance in
Sec.~\ref{sec:shallow-performance}.

\subsection{Complete exhaustive campaign}

For random connected 3-regular MaxCut, the discovery campaign uses
$n\in\{10,12,14,16,18\}$, 100 graphs per $(n,p)$ cell, and
$p=1,\ldots,5$.  The deterministic QRAC coloring produces
$r=4,\ldots,8$.  At each graph and depth the selector enumerates every one of
the $4^r$ states in the $2^r$ noncomputational MUB families, with $2^r$ labels
per family.  All candidates and the baseline receive the same angle optimizer.

The authenticated completed merge contains 20,308,480 state evaluations and
2,500 baseline evaluations.  Every planned $(n,p)$ cell contains 100 complete
graph scans.  Excluding the historical $r=8$ nonproduct data leaves 9,822,720
validated state rows, 2,340 graph scans, and 20 validated $(r,p)$ aggregations
for $r=4,\ldots,7$ and $p=1,\ldots,5$.  Additional implementation and validation details appear in
Appendix~\ref{app:experiments}.

\subsection{Product-family rank across register and depth}

For each graph and family we first average decoded ratio over all $2^r$
labels, as in Eq.~\eqref{eq:family-mean}.  At fixed $(r,p)$ we then average
over graphs, identify the strongest nonproduct family, and form paired graph
differences.  Figure~\ref{fig:family-discovery} reports the exact graph-first
pattern; the complete 20-row numerical result is retained in
Appendix~\ref{app:experiments}, Table~\ref{tab:family-ranks}.

\begin{figure*}[t]
  \centering
  \includegraphics[width=\textwidth]{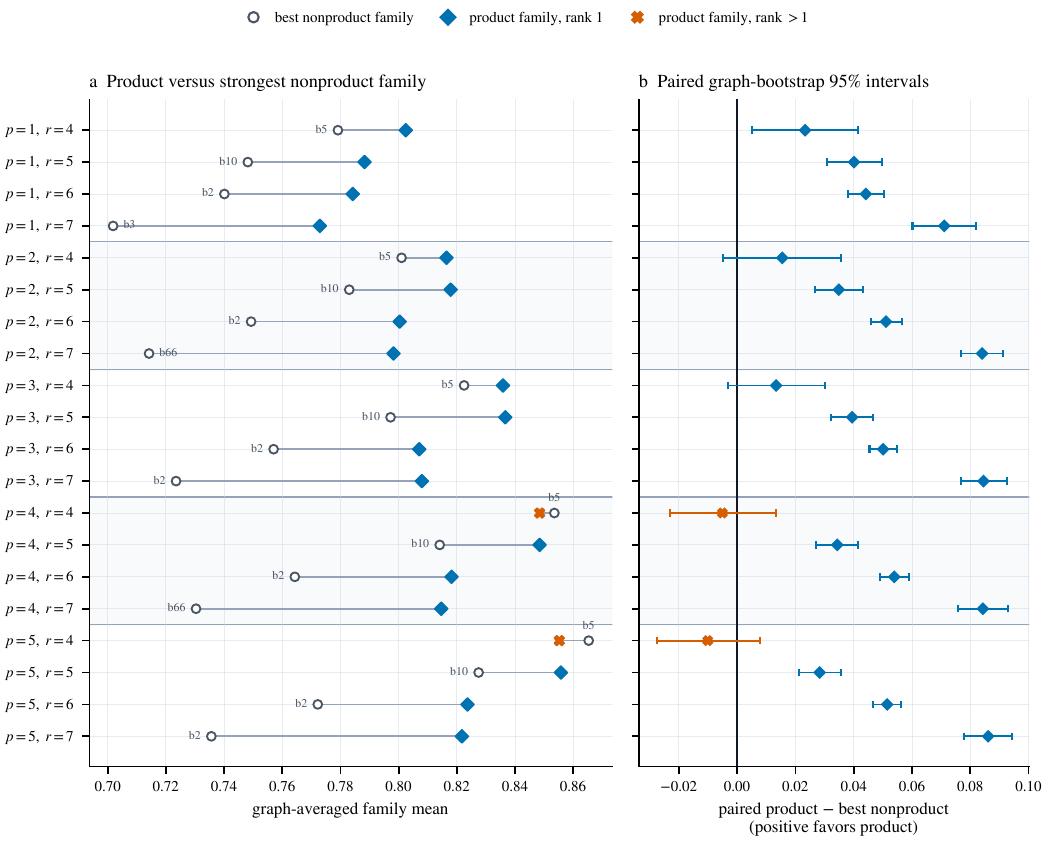}
  \caption{Exhaustive MaxCut family discovery in the 20 validated full-campaign
  aggregations.  \textbf{a}, product-family and selected strongest-nonproduct
  graph-mean decoded ratios, with the nonproduct basis labeled directly.
  \textbf{b}, paired product-minus-nonproduct margin with graph-bootstrap 95\%
  intervals (10,000 resamples) and a zero reference.  Every row enumerates all
  noncomputational families and all $2^r$ labels per family.  The intervals
  describe the selected same-sample comparator and are not simultaneous,
  selection-adjusted intervals.  Historical $r=8$ nonproduct data are excluded
  because they predate the finite-field correction.}
  \label{fig:family-discovery}
\end{figure*}

The product family ranks first in 18 of the 20 aggregations.  At the smallest
register, basis 5 leads at $p=4$ and $p=5$; the product family is second, with
paired margins $-0.0051$ and $-0.0101$.  At $r=5,6,7$, the product family
ranks first at every tested depth, including $p=5$.  Across its 18 leading
cells, the paired margin ranges from $0.0134$ to $0.0861$.  The strongest
nonproduct changes among bases 2, 3, 5, 10, and 66.  Equal family cardinality
and a common optimizer remove candidate count as an explanation for these
family means.

At $r=4$ and $p=2,3,4,5$, the descriptive intervals cross zero,
including the two cells in which basis 5 leads.  No ordering is inferred in
those cells.  We therefore report the finite-campaign ordering---18 rank-one
cells and persistence through depth five at $r=5,6,7$---without claiming a
universal product-family law.

\subsection{Why historical \texorpdfstring{$r=8$}{r=8} comparisons are excluded}
\label{sec:gf2-defect}

An earlier finite-field power-table recurrence omitted reduction modulo two
when multiple pentanomial taps contributed to one coefficient.  The defect
affects $r=8$ but not $r=4,5,6,7$.  The repaired catalog restores the MUB
overlap invariant.  Family 1 ($j=0$) is algebraically immune, but 248 of the
256 noncomputational $r=8$ bases change, so the historical all-family scan
cannot compare the product family with the corrected nonproduct catalog.  No
corrected exhaustive rerun is available.  Post-fix product-family scale data
at $r\ge8$ remain valid; historical $r=8$ nonproduct ranks and brightness
maps are excluded.

The exhaustive catalog therefore serves as a discovery instrument: it reveals
the unique product-$X$ family as the stable leading structure on the unaffected $r=5,6,7$ registers.  That finding motivates \SEMIQAOA{}'s restriction to the
product family and the two polynomial-size random budget schedules tested by
\RandSEMIQAOA{}.  The cross-family $\Xi_{123}$ pool remains a structured MUB
comparator rather than a \SEMIQAOA{} candidate set.  The discovery does not
imply that every product subset or every register--depth cell must win.

\section{Depth-one dynamical lift}
\label{sec:dynamics}

High best-of-$K$ performance could arise either because the candidate states
already decode to favorable cuts or because the cost--mixer evolution changes
their decoded observables.  A zero-angle control and an exact finite-population
order-statistic calculation separate these effects within the product family
used by both tested \RandSEMIQAOA{} budget schedules.

\subsection{Raw states versus depth-one optimization}

The control uses 72 graphs spanning nine $(n,r)$ cells.  For $r\le8$ it
evaluates all $2^r$ product labels; at $r=9,10$ it uses a fixed uniform sample
of 64.  ``Raw'' means sign decoding at $\gamma=\beta=0$, verified against
zero-angle circuit evolution.  ``Post-$p=1$'' means the same labels after the
campaign angle optimizer.

\begin{figure*}[t]
  \centering
  \includegraphics[width=\textwidth]{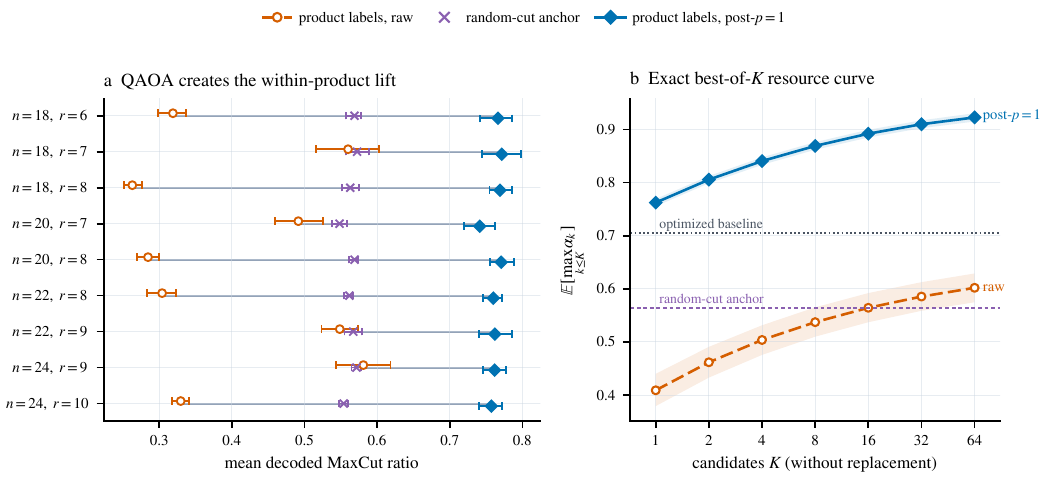}
  \caption{No-QAOA and candidate-count controls.  \textbf{a}, cell-level mean
  decoded ratios for raw product labels, the random-cut anchor, and the same
  labels after depth-one optimization.  Labels give $n/r$; every cell contains
  eight graphs.  \textbf{b}, exact expected best decoded ratio from $K$
  product labels sampled uniformly without replacement.  Equation~\eqref{eq:best-k-exact}
  is evaluated within each graph before graph averaging; ribbons are 95\%
  graph-bootstrap intervals.  Horizontal references show the random-cut
  anchor and optimized $\ket{+}^{\otimes r}$ QRAO baseline.}
  \label{fig:dynamical-origin}
\end{figure*}

Across the nine equally weighted cells, the mean product-family ratio is
$0.4090$ before QAOA, below the $0.5640$ random-cut anchor, and $0.7623$ after
depth-one optimization.  The mean within-product lift is therefore $0.3533$.
For graph $G$, the anchor is $(|E(G)|/2)/C_{\max}(G)$---the expected cut of a
uniform random bitstring divided by the exact maximum---and it is averaged
graph first within each cell.
For every even-$r$ cell, the raw $\ket{+}^{\otimes r}$ sign decode is zero on
all sampled graphs, while its optimized ratio lies between $0.67$ and $0.73$.
This parity artifact reinforces that a raw deterministic sign decoder is not
a proxy for variational response.

Figure~\ref{fig:dynamical-origin}b controls candidate count without Monte
Carlo sampling.  At $K=1$ the curves reproduce the raw and optimized means.
At $K=64$, the expected maximum is $0.602$ raw and $0.923$ post-$p=1$.  More
candidates help both, but the raw ensemble remains far below the optimized
one.  Thus the $>0.9$ best-of-$K$ regime is not obtained merely by decoding 64
classically lucky product labels.  This control establishes a within-product
dynamical lift; because it does not repeat the raw/post experiment for every
entangled family, it does not by itself explain the between-family ordering
in Fig.~\ref{fig:family-discovery}.

The curves use all stored product-family labels with $K\le64$, so all 72
graphs contribute at every displayed point. For $r=9,10$, only 64 labels were
stored; hence no claim is made about the full $512$- or $1{,}024$-state
product families.

Equation~\eqref{eq:best-k-exact} is also the exact finite-population
candidate-budget curve underlying \RandSEMIQAOA{} when $K$ product labels are
sampled uniformly without replacement.  It averages over all such subsets;
it is not the realized score of the one frozen experimental pool at each
budget.

\subsection{Relaxed energy and decoded score are not interchangeable}

The circuit angles are optimized against relaxed QRAO energy, whereas the
headline statistic selects the largest decoded MaxCut ratio.  We regenerated
their alignment with the same r-specific candidate-union completeness gate as
the primary scale analysis.  Of 12,449 stored graph--depth scans, 17 incomplete
unions are excluded, leaving 12,432 scans and 49,728 complete primary-arm
records whose arm maxima agree bit-for-bit with the main reduction.
The gate covers the executed physical union, including the legacy $R_E$
contribution described in Appendix~\ref{app:experiments}; it is not restricted
to the displayed scientific arms.  The frozen analysis archive retains exact per-scan
sufficient statistics rather than raw candidate-state rows.  Accordingly this
diagnostic describes objective mismatch in the executed campaign; it does not
validate the unresolved historical $r=8$ component of $R_E$.

Across complete scans, the within-scan Spearman correlation has mean $0.5285$
and median $0.5313$.  The relaxed-energy argmax is also a decoded-score
argmax, including ties, in $22.27\%$ of scans; selecting it instead of the
decoded argmax loses $0.06991$ in decoded ratio on average.  Giving each
$(n,p)$ cell equal weight yields $23.11\%$ and $0.06939$, respectively.  These
moderate correlations explain why energy optimization can improve the pool
without making energy a reliable state-selection oracle.  All reported
best-of-set results therefore use the stored decoded score after each
candidate's angle optimization.  The reproducible gate, source digests, and
cell-level values are given in Appendix~\ref{app:experiments}.

This mismatch also fixes the theoretical scope.  Section~\ref{sec:mechanism}
controls common-angle, label-averaged relaxed-energy response.  It does not
control independently optimized decoded maxima, and it is presented as a
mechanism consistent with the observations rather than a theorem for the
$0.94$ result.

\subsection{Optimizer and decoder checks}

A 30-graph local-iteration sweep reoptimizes the already selected
$\Xi_{123}$ winner at $p=1,3,5$.  Decoded ratios stabilize by 10, 40, and 80
L-BFGS-B iterations, respectively; the campaign budget is 80.  This validates
the local budget for those selected states, not global optimality over states
or angles.  A separate tight $p=5$ stress test improves relaxed energy on all
20 graphs but worsens decoded ratio on 13, consistent with the alignment
diagnostic.

An auxiliary expectation-level magic-decoder score preserves the ordering of
the three primary arms in the retained analysis, but it is an analytic
expectation rather than a finite-shot rounding protocol.  We therefore do not
claim decoder independence.  All main ratios use the deterministic Pauli-sign
decoder in Eq.~\eqref{eq:decoded-ratio}; device noise, measurement cost, and
finite-shot selection remain outside this statevector study.

\section{One-layer relaxed-energy mechanism}
\label{sec:mechanism}

This section studies the family-averaged one-layer relaxed-energy response, not
the decoded ratio.  It has two parts.  The deterministic result identifies an extensive
edge-syndrome code and proves dephasing of each nontrivial mixer branch.  The
Gaussian result controls their sum on the natural entangled-family mixer scale
and separates it from a fixed-angle product response.  Complete proofs appear
in Appendices~\ref{app:det-proofs} and \ref{app:gaussian-proofs}.

There is a further order-of-operations distinction.  The theoretical response
uses one common angle pair $(\gamma,\beta)$ for every label $v$ and averages
over $v$ before any supremum is taken.  The experimental family mean instead
optimizes an angle vector independently for every state and then averages its
decoded score.  In symbols, the theory can control expressions of the form
$\sup_{\gamma,\beta}|2^{-r}\sum_v f_{j,v}(\gamma,\beta)|$, not
$2^{-r}\sum_v\sup_{\gamma,\beta}f_{j,v}(\gamma,\beta)$.  Label cancellation
can make the former small while the latter remains large; the theorem is a
mechanism, not the theoretical counterpart of Eq.~\eqref{eq:family-mean}.

\begin{figure*}[t]
  \centering
  \includegraphics[width=\textwidth]{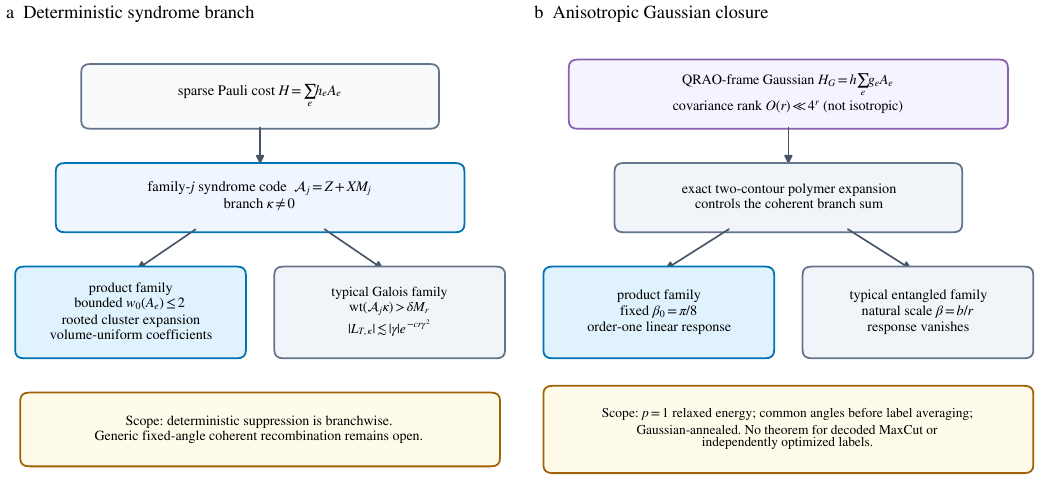}
  \caption{One-layer relaxed-energy mechanism.  \textbf{a}, in the
  deterministic expansion a typical entangled-family branch induces an
  extensive edge-sign quench and its individual rooted overlap dephases;
  product-family terms admit a rooted connected-cluster expansion with
  volume-uniform factorial-weighted bounds.  Coherent
  recombination of all deterministic branches at generic fixed mixer angle is
  open.  \textbf{b}, the anisotropic Gaussian QRAO surrogate closes that sum
  on the entangled scale $\beta=b/r$, while the product family retains a
  nonzero linear response at $\beta_0=\pi/8$.  The comparison is annealed,
  common-angle, label-averaged, $p=1$, and concerns relaxed energy rather than
  independently optimized decoded cuts.}
  \label{fig:theory-mechanism}
\end{figure*}

\subsection{Stabilizer filter and syndrome map}

Write the traceless local cost as
\begin{equation}
 H=\sum_{A\in\calH_r}h_AA,\qquad M_r=|\calH_r|=\Theta(r),
 \label{eq:local-pauli-cost}
\end{equation}
where identical Pauli strings have first been combined, so the Hermitian
Paulis $A=P(x_A,z_A)$ are pairwise distinct, and the physical term-overlap
graph has uniformly bounded degree.  Let
$\tau_r(O)=2^{-r}\Tr O$.

\begin{lemma}[Stabilizer filter]
\label{lem:stabilizer-filter}
For every operator $O=\sum_{x,z}\widehat O(x,z)P(x,z)$,
\begin{equation}
 \langle\psi_{j,v}|O|\psi_{j,v}\rangle
 =\sum_x \widehat O(x,M_jx)(-1)^{v\cdot x}.
 \label{eq:stabilizer-filter}
\end{equation}
Consequently,
\begin{equation}
 2^{-r}\sum_v\langle\psi_{j,v}|O|\psi_{j,v}\rangle=\tau_r(O).
 \label{eq:family-trace}
\end{equation}
\end{lemma}

For a Pauli $A$, define its family-$j$ syndrome
\begin{equation}
\begin{aligned}
 a_j(A)&=z_A+M_jx_A\in V,\\
 w_j(A)&=|\supp a_j(A)|.
\end{aligned}
\label{eq:syndrome}
\end{equation}
Its support is exactly the set of matched mixer generators that anticommute
with $A$.  A basic spread property will repeatedly randomize this syndrome.

\begin{lemma}[Exact Galois randomization]
\label{lem:galois-randomization}
For every nonzero $\kappa\in V$, the map
$j\mapsto M_j^{\mathsf T}\kappa$ is a bijection on $V$.
\end{lemma}

\subsection{Exact branch representation}

For fixed $(j,v)$ and $\kappa\subseteq[r]$, identify $\kappa$ with its indicator vector in $V$ and put
$G_\kappa=\prod_{q\in\kappa}g_{j,v,q}$ and
$H^{(\kappa)}=G_\kappa HG_\kappa$.  For a root term $T$, let
$S_T=\supp a_j(T)$ and define
\begin{equation}
 L_{T,\kappa}(\gamma)=\tau_r\!\left(
 e^{\ii\gamma H^{(\kappa)}}T e^{-\ii\gamma H}\right).
 \label{eq:branch-amplitude}
\end{equation}
Let
\begin{equation}
 U_{j,v}(\gamma,\beta)
 :=e^{-\ii\beta B_{j,v}}e^{-\ii\gamma H}.
\end{equation}
Define the normalized, common-angle, family-averaged energy and its change by
\begin{align}
 f_{j,r}(\gamma,\beta)
 &:=\frac{1}{M_r2^r}\sum_{v\in V}
 \langle\psi_{j,v}|U_{j,v}^{\dagger}HU_{j,v}|\psi_{j,v}\rangle,
 \nonumber\\[-2pt]
 \Delta f_{j,r}(\gamma,\beta)
 &:=f_{j,r}(\gamma,\beta)-f_{j,r}(0,\beta).
 \label{eq:family-response}
\end{align}

\begin{theorem}[Exact matched-mixer branches]
\label{thm:exact-branches}
With $c_\beta=\cos(2\beta)$ and $s_\beta=\sin(2\beta)$,
\begin{equation}
\begin{aligned}
 \Delta f_{j,r}
 &=\frac1{M_r}\sum_T h_T
 \sum_{\varnothing\ne\kappa\subseteq S_T}\\
 &\quad\times
 c_\beta^{w_j(T)-|\kappa|}(-\ii s_\beta)^{|\kappa|}
 L_{T,\kappa}(\gamma).
\end{aligned}
\label{eq:exact-branch-formula}
\end{equation}
The empty branch vanishes identically.  In particular,
$\Delta f_{j,r}(\gamma,k\pi/2)=0$ for every integer $k$.
\end{theorem}

Equation~\eqref{eq:exact-branch-formula} is the central bookkeeping identity.
It makes two different tasks visible.  One may suppress each
$L_{T,\kappa}$, or one may control the coherent radial sum over $\kappa$.
The former does not automatically imply the latter.

\subsection{An edge-syndrome code for typical Galois families}

Index the Hamiltonian terms by $A=1,\ldots,M_r$, and let $X,Z$ be the matrices
with rows $x_A^{\mathsf T},z_A^{\mathsf T}$.  Define
\begin{equation}
\begin{aligned}
 \mathcal A_j&=Z+XM_j, & \Phi&=[Z\ \ X],\\
 d_r&=2r-\rank\Phi.
\end{aligned}
\label{eq:syndrome-code}
\end{equation}
Then $(\mathcal A_j\kappa)_A=\kappa\cdot a_j(A)$ is the indicator that branch
$\kappa$ flips Hamiltonian term $A$.  Let
\begin{equation}
 H_2(x)=-x\log_2x-(1-x)\log_2(1-x).
 \label{eq:binary-entropy}
\end{equation}
For $0<x<1$, this is the binary entropy.

\begin{theorem}[Positive distance for typical families]
\label{thm:code-distance}
Suppose $d_r\le d_0$ and choose $0<\delta<1/2$ such that
\begin{equation}
 M_rH_2(\delta)+d_0\le(1-\eta)r
 \label{eq:entropy-condition}
\end{equation}
for some $\eta>0$.  All but at most
$2^{(1-\eta)r+o(r)}$ Galois families satisfy
\begin{equation}
 \wt(\mathcal A_j\kappa)>\delta M_r
 \quad\text{for every }\kappa\ne0.
 \label{eq:positive-distance}
\end{equation}
\end{theorem}

The abstract rank hypothesis is automatically bounded for the QRAC packing
used in the experiments.

\begin{proposition}[Rank and locality of the cubic-graph packing]
\label{prop:actual-packing}
For a connected cubic graph encoded by exact coloring and color-class chunks
of size at most three, $d_r\le5$, $M_r=\Theta(r)$, and the term-overlap degree
is at most 16.  Moreover, asymptotically at least a $2/3$ fraction of edge
terms have nonzero $X$ label, and at least a $2/3$ fraction have nonzero
product-family syndrome.
\end{proposition}

The constants are deliberately conservative.  The proof uses only connected
cubic incidence rank, the three distinct Pauli axes in a full chunk, and the
four-color bound; it does not depend on a random-graph model.

\subsection{Deterministic individual-branch dephasing}

Let
\begin{equation}
 Q_j(\kappa)=\sum_{A:(\mathcal A_j\kappa)_A=1}h_A^2.
\end{equation}
An exact two-contour connected expansion factors
$L_{T,\kappa}=Z_\kappa R_{T,\kappa}$, with
\begin{equation}
 Z_\kappa=\tau_r(e^{\ii\gamma H^{(\kappa)}}e^{-\ii\gamma H}).
\end{equation}

\begin{theorem}[Uniform deterministic branch dephasing]
\label{thm:det-dephasing}
Assume $|h_A|\le h_*$, bounded term-overlap degree, and
$Q_j(\kappa)\ge q_*r$.  There exist constants $c,C,\gamma_*>0$, independent
of $r,j,T,\kappa$, such that
\begin{equation}
 |L_{T,\kappa}(\gamma)|
 \le C|\gamma|e^{-cr\gamma^2},
 \qquad |\gamma|\le\gamma_*.
 \label{eq:det-dephasing}
\end{equation}
For equal-magnitude QRAO edge coefficients, Eq.~\eqref{eq:det-dephasing}
holds simultaneously for every nonempty branch of every family satisfying
Theorem~\ref{thm:code-distance}.
\end{theorem}

The theorem proves a deterministic mechanism: a typical entangled family
turns every nonempty branch into an extensive sign quench, whose two-contour
overlap decays.  It does \emph{not} prove that the coherent sum in
Eq.~\eqref{eq:exact-branch-formula} is exponentially small at every fixed
$\beta$.  Exact resonances are controlled, but a generic-bulk cancellation
theorem remains open.

\subsection{Why isotropic Gaussian disorder cannot select a family}

Let $G$ be a centered Gaussian operator.  If its covariance is proportional
to the identity on traceless Hermitian operator space, conjugation by the
Clifford relating any two MUB families preserves its law.  Consequently the
entire matched-QAOA response process has the same distribution in every
family.  An isotropic surrogate is therefore a no-go model for the empirical
product-family preference.

The QRAO surrogate retains the actual sparse Pauli frame.  Relabel the
$M_r$ cost terms by $e\in E_r$ and set
\begin{equation}
 H_G=h\sum_{e\in E_r}g_eA_e,\qquad g_e\stackrel{\mathrm{iid}}{\sim}N(0,1).
 \label{eq:gaussian-qrao}
\end{equation}
With the Hilbert--Schmidt inner product, its covariance operator is
$\Sigma_G:X\mapsto h^2\sum_e A_e\Tr(A_eX)$.  Its rank is at most $M_r=O(r)$ inside an operator
space of dimension $4^r-1$, so it is exponentially anisotropic.

\subsection{Exact Gaussian polymer expansion and separation}

Let $\mathcal F_{j,r}(\gamma,\beta)$ denote the normalized annealed
label-averaged energy change defined explicitly in
Eq.~\eqref{eq:gaussian-response-definition}.  For branch $\kappa$, set
$\sigma_e(\kappa)=(-1)^{(\mathcal A_j\kappa)_e}$ and
$H_\kappa=h\sum_e\sigma_e g_eA_e$.  Define the annealed vacuum
\begin{equation}
 Z_\kappa(\gamma)=\E_g\tau_r(
 e^{\ii\gamma H_\kappa}e^{-\ii\gamma H_G}).
 \label{eq:annealed-vacuum}
\end{equation}
Grouping two-contour words by their exact set of Gaussian labels yields a
hard-core gas of connected edge sets.  It is exact: no commutation
approximation is made inside a polymer, and Gaussian averaging retains all
Wick pairings.

\begin{theorem}[Gaussian polymer control]
\label{thm:polymer-control}
Let $D$ be one plus the maximum term-overlap degree and
$u_{\rm abs}(\gamma)=e^{2h^2|\gamma|^2}-1$.  A connected polymer $C$ has activity
\begin{equation}
 |W_\kappa(C;\gamma)|\le u_{\rm abs}(\gamma)^{|C|}.
\end{equation}
If $u_{\rm abs}(\gamma)\le(16e^2D^2)^{-1}$, the polymer and logarithmic expansions
converge uniformly.  If $N_-(\kappa)=|\{e:\sigma_e=-1\}|$, then on a smaller
uniform disk
\begin{equation}
 \log Z_\kappa=-2h^2N_-(\kappa)\gamma^2
 +O(M_r\gamma^4).
 \label{eq:singleton-resummation}
\end{equation}
The corresponding rooted branch is bounded by
\begin{equation}
 C|h\gamma|\exp[-2h^2N_-(\kappa)\gamma^2+C_4M_r\gamma^4].
 \label{eq:gaussian-root-bound}
\end{equation}
\end{theorem}

For fixed nonzero $\kappa$, Lemma~\ref{lem:galois-randomization} makes
$M_j^{\mathsf T}\kappa$ exactly uniform as $j$ varies.  Bounded differences
then show that an extensive number of active edges flip simultaneously for
all branches of weight at most $\lfloor\log r\rfloor$, outside an
exponentially small family set.  The factorial tail of the mixer expansion
controls larger branches when $\beta=b/r$.

\begin{theorem}[Typical-entangled Gaussian dephasing]
\label{thm:gaussian-dephasing}
Assume $M_r=\Theta(r)$, bounded overlap, and
$\liminf_{r\to\infty}M_r^*/M_r\ge p_*>0$, where $M_r^*$ counts terms with
nonzero $X$ label.  Fix $B=[-b_+,b_+]\subset\R$.  There are constants
$C,C_4,c_G,\gamma_*>0$ and a deterministic sequence $\eta_r\to0$ such that,
for all but an $e^{-c_Gr+o(r)}$ fraction of nonzero Galois families,
\begin{equation}
 |\mathcal F_{j,r}(\gamma,b/r)|
 \le C|\gamma|e^{-M_r(c_2\gamma^2-C_4\gamma^4)}+\eta_r
 \label{eq:gaussian-family-bound}
\end{equation}
uniformly for $|\gamma|\le\gamma_*$ and $b\in B$, where
$c_2=p_*h^2/2$.  After reducing $\gamma_*$ if necessary,
\begin{equation}
 \sup_{|\gamma|\le\gamma_*\!,\ b\in B}
 |\mathcal F_{j,r}(\gamma,b/r)|\longrightarrow0.
 \label{eq:gaussian-vanish}
\end{equation}
\end{theorem}

For the product family, define
\begin{equation}
 s_{0,r}(\beta)=\frac1{M_r}\sum_{e\in E_r}
 \sin\!\left(2\beta w_0(A_e)\right).
 \label{eq:product-slope-density}
\end{equation}
The sum is over \emph{all} edge terms.  In particular, $ZZ$ edges were
incorrectly omitted by an earlier ``nonzero $X$ label'' shorthand even though
they have product-family syndrome weight two.

\begin{theorem}[Product response and natural-scale separation]
\label{thm:product-separation}
Assume bounded overlap and a uniform product-syndrome bound
$w_0(A_e)\le w_*$.  Then
\begin{equation}
 \mathcal F_{0,r}(\gamma,\beta)
 =-2h^2s_{0,r}(\beta)\gamma+O(|\gamma|^3)
 \label{eq:product-linear-response}
\end{equation}
uniformly in $r$ on a fixed small disk.  For the cubic QRAC packing $w_*=2$
and $\beta_0=\pi/8$,
\begin{equation}
 \liminf_{r\to\infty}s_{0,r}(\beta_0)\ge\frac{\sqrt2}{3}>0.
 \label{eq:product-slope-liminf}
\end{equation}
Consequently, there exist a sufficiently small fixed $\gamma_0<0$ and
$c_0>0$ such that
$\mathcal F_{0,r}(\gamma_0,\beta_0)\ge c_0$ for all sufficiently large $r$.
By Theorem~\ref{thm:gaussian-dephasing}, the typical entangled-family response
on $\beta=b/r$ vanishes uniformly, so the response difference is bounded below
by a positive constant for all sufficiently large $r$.
\end{theorem}

The scale qualification is essential.  Theorem~\ref{thm:product-separation}
compares fixed product angle $\beta_0$ with the natural entangled scale
$b/r$; it does not optimize an entangled family over all $\beta$.  It is also
annealed over Gaussian edge coefficients.  Its value is mechanistic: sparse
QRAO covariance plus Galois syndrome spreading suppress typical entangled
branches, while bounded-weight product syndromes preserve a coherent linear
response.

\section{Discussion}
\label{sec:discussion}

On the complete $n=18,20,22$ cohort, matched-budget \RandSEMIQAOA{} attains
mean decoded ratio $0.9421$ with a compressed depth-one best-of-$K$ search.
The structured cross-family MUB selector $\Xi_{123}$ is slightly higher at
$0.9439$, while the $r^2$-budget \RandSEMIQAOA{} result is $0.9319$.  No
problem-specific classical approximate cut initializes the quantum state, but
the experiment uses substantial hybrid resources: $K$ fixed state--mixer
candidates, an independent angle search for each candidate, and decoded-score
selection.  The exhaustive MUB campaign explains where the two random
product-family budgets should be concentrated.  Its product family leads 18
of 20 validated family-mean cells and every tested depth at $r=5,6,7$; the no-QAOA
control shows, within the product family, that the raw-to-optimized lift is not
inherited from favorable mean raw cuts.

\subsection{Different routes to shallow performance}

Depth alone is not a complete resource measure.  Three architectures illustrate
how work can be moved outside the number of QAOA alternations.

\emph{Parameter expressivity.}  Multi-angle QAOA assigns separate cost and
mixer angles to edges and vertices within one alternation
\cite{Herrman2022MultiAngle}.  It reports mean expected approximation ratio
$0.9257$ across all connected nonisomorphic eight-vertex graphs, while its
mean expected-energy ratios on triangle-free random 3-regular graphs are
$0.8123$ at $n=50$ and $0.8000$ at $n=100$.  Its resource is an enlarged,
instance-optimized parameter vector rather than a candidate ensemble.

\emph{Classical-solution-informed initialization.}  Regularized warm-started
QAOA constructs an instance-specific continuous classical solution that sets
the initial single-qubit states, followed by a shallow quantum layer and
rounding \cite{He2026RWSQAOA}.  Its Fig.~2 reports measured-bitstring approximation ratios near $0.97$
without local search at $p=1$ on five $n=96$ random 3-regular hardware
instances; the single-step local-search variant is reported separately.  That performance belongs to a deliberately different architecture:
the quantum state is supplied with a strong instance-specific classical warm
start.

\emph{\SEMIQAOA: compressed product-family ansatz search.}  Here a
$(3,1)$-QRAC compresses $n$ classical variables to $r$ qubits, and the ansatz
search is restricted to product-$X$ MUB states and their matched signed-$X$
mixers.  \RandSEMIQAOA{} offers $K$ labels drawn uniformly without replacement
from that family.  Every offered ansatz is independently optimized at the same
$p$, then selected by decoded score.  \SEMIQAOA{} should therefore be viewed
as a third resource tradeoff: rather than increasing alternation depth or
supplying an instance-specific classical solution, it searches a compressed
product-family state--mixer space.  The cost is multiple quantum
optimizations, and the approach is distinct rather than uniformly cheaper.

Other shallow quantum-assisted pipelines move still different work outside the
direct circuit output.  Quantum relax-and-round diagonalizes and rounds a
global correlation matrix \cite{DupontSundar2024QRR}; quantum-informed
surrogate sampling constructs a classical factor model and samples it from
low-order quantum correlators \cite{WyboFinzgar2026QISS}; and recursive QAOA or
QRAO repeatedly invokes shallow quantum information while eliminating
variables \cite{Finzgar2024QIRO,KondoEtAl2025}.  These methods are relevant
comparators but their expected energies, rounded solutions, surrogate samples,
and recursive final cuts cannot be placed on one numerical leaderboard with
our decoded best-of-$K$ statistic.

Prior work applying alternating operators directly to QRAO compared matched
$X$, $Y$, and $Z$ product state--mixer pairs and transferable fixed parameters
on random 3-regular MaxCut \cite{HeEtAl2025QRAO}.  The present work instead
enumerates a complete matched stabilizer-MUB catalog with statewise angle
optimization, uses that exhaustive result to identify a finite-budget pool, and
develops syndrome and Gaussian mechanisms for its family dependence.

\subsection{What the product geometry and decoder diagnostics say}

The structured cross-family MUB selector $\Xi_{123}$ is a useful
discovery-derived pool and has the slightly larger complete-cohort point
estimate, $0.9439$ versus $0.9421$ for matched-budget \RandSEMIQAOA{}.
Nevertheless, the random product-family method is the cleaner flagship: it
stays entirely inside the discovered product family, removes the fitted label
geometry in $\Xi_{123}$, and leads from $r=9$ in the one-seed depth-one
equal-$K$ comparison.  The $r^2$-budget \RandSEMIQAOA{} result, $0.9319$,
shows the same algorithmic mechanism under a smaller polynomial budget.
Together with the exact best-of-$K$ control, these comparisons suggest that
access to many responsive product labels matters more than the particular
tested index geometry.  With only one random subset per $r$ and schedule, the
appropriate next experiment is a frozen multi-seed selector study rather than
another post hoc region fit.

Relaxed energy and decoded ratio are also distinct resources and objectives.
Their median within-scan Spearman correlation is only $0.5313$, and the energy
argmax is a decoded argmax in $22.27\%$ of completeness-gated scans.  A
relaxed-energy response theorem can therefore explain structural sensitivity
without certifying the nonlinear sign decoder or the observed $0.94$
best-of-set ratio.

\subsection{Limitations and falsifiable next steps}

The empirical headline is a finite-size, noiseless statevector result.  Its
clean inferential population is the 2,400 complete $n=18,20,22$ graph scans,
conditional on fixed candidate rules and one random subset seed per register.
The $n\ge28$ points are runtime-censored toward smaller completed registers.
The observed matched-budget \RandSEMIQAOA{} mean is $0.9147$ on the 65
completed $n=34$ scans, while the $r^2$-budget \RandSEMIQAOA{} mean is
$0.9026$ and the structured MUB selector $\Xi_{123}$ is $0.8946$.  Thus both
tested random product-family schedules have point means above $0.9$ in that
analysis freeze.  Those tail points are descriptive rather than unbiased population
estimates.

The candidate resource is material.  The complete-cohort arms contain
$K=36$--232 states, while the executed deduplicated union plus baseline costs
113--415 independent optimizations per graph.  Each uses a large grid and
three local starts at $p=1$.  We do not compare wall-clock time, measurement
cost, or total classical work with standard QAOA, ma-QAOA, warm-started
methods, recursive solvers, or classical MaxCut algorithms.  The observed
ratio is neither a worst-case approximation guarantee nor evidence of quantum
advantage.

Within the product-family restriction that defines \SEMIQAOA, the
label-selection rule remains a separate design choice.  \RandSEMIQAOA{} uses
uniform random selection; future work can investigate problem-informed
alternatives to uniform random product-label selection.

All circuits are simulated noiselessly, and candidate selection uses the
stored deterministic Pauli-sign decoded score.  Device noise, finite-shot
ranking, and the state-preparation and measurement schedules require separate
study.  Repeated random candidate subsets are needed to estimate
selector-to-selector uncertainty.

The exhaustive result is also bounded by the historical $r=8$ finite-field
defect.  A corrected nonproduct rerun can test whether product-family
leadership persists at the next register.  Finally, the strongest mechanism
theorem is one-layer, annealed over Gaussian coefficients, and compares
typical entangled families at $\beta=b/r$ with the product family at fixed
$\beta_0$.  The deterministic proof controls individual branches but leaves
generic coherent recombination open.  None of these results proves decoded
ordering or $p>1$ performance.

\section{Conclusion}
\label{sec:conclusion}

\RandSEMIQAOA{} samples labels uniformly without replacement from the
product-$X$ stabilizer-MUB family and independently optimizes the corresponding
matched QRAO--QAOA circuits.  On 2,400 complete random connected 3-regular
MaxCut instances at $n=18,20,22$, the matched and $r^2$ candidate schedules
attain mean decoded best-of-set ratios $0.9421$ and $0.9319$, respectively, at
$p=1$; the structured cross-family comparator $\Xi_{123}$ attains $0.9439$.
Both random product-family schedules remain above $0.9$ on the completed
$n=34$ scans, although the $n\ge28$ tail is runtime-censored and only
descriptive.

The exhaustive MUB campaign explains the restriction: the product family leads
18 of 20 validated family-mean cells and every tested depth for $r=5,6,7$.  The
zero-angle control shows that high best-of-set performance is generated by
QAOA dynamics within this family rather than by unusually strong raw decoded
states.  The deterministic syndrome argument and anisotropic Gaussian polymer
expansion supply a one-layer relaxed-energy mechanism, not a theorem for the
nonlinear decoder.  The result is therefore a resource-explicit empirical
route to high shallow performance; it is not a quantum-advantage claim or a
comparison of total computational cost.

\newpage

\section*{Data and code availability}
\label{sec:data-code-availability}
The frozen configurations, selector seeds, optimization settings, and
data-processing procedures are described in
Appendix~\ref{app:experiments}.  The corresponding code, numerical records,
and derived tables are available from the authors and will be deposited in a
versioned public archive.

\printbibliography

\clearpage
\onecolumn
\appendix
\section{Experimental methods and additional checks}
\label{app:experiments}

\subsection{Graph ensembles and QRAC assignment}

Every main experiment uses simple connected 3-regular graphs.  For a requested
size $n$ and graph index $g$, the generator tries seeds
$s_0+100g,s_0+100g+1,\ldots$ until NetworkX returns a connected graph, then
relabels vertices in sorted order.  Discovery uses $s_0=5000$ and scale uses
$s_0=7000$; these are disjoint graph cohorts even where $n=18$ appears in both
campaigns.  A graph identifier records $n$, the index, and the accepted seed.

The $(3,1)$-QRAC assignment is deterministic.  Exact backtracking finds the
minimum proper coloring, checking color counts in ascending order and vertices
and colors in ascending label order.  Each color class is split into chunks of
at most three vertices.  If a chunk is assigned to relaxed qubit $q$, its
positions receive the cyclically shifted axes
$X,Y,Z$ in the order $(s+q)\bmod3$.  A color class is an independent set, so
no graph edge has both endpoints on one relaxed qubit.  The complete
assignment, including the variable-to-qubit and variable-to-axis maps, is
hashed into every output row.

For unweighted MaxCut,
\begin{equation}
 C(z)=\frac{|E|}{2}-\frac12\sum_{(i,j)\in E}z_iz_j,
 \qquad
 H=\frac{|E|}{2}I-\frac32\sum_{(i,j)\in E}P_iP_j.
 \label{eq:maxcut-qrao}
\end{equation}
The identity is retained during numerical optimization but omitted from the
response theory.  Exact denominators are obtained by exhaustive bitstring
evaluation for $n\le22$.  Above this threshold a deterministic single-worker
CP-SAT model is required to return status \texttt{OPTIMAL}; otherwise the run
fails rather than reporting a ratio with an uncertified denominator.

\subsection{MUB enumeration and invariants}

Implementation basis identifiers are zero for the computational basis and
$1,\ldots,d$ for the Galois noncomputational bases.  The state identifier is
\begin{equation}
 \texttt{mub\_state\_id}=d\,\texttt{basis\_id}+\texttt{vector\_id}.
\end{equation}
The exhaustive selector emits every
$\texttt{basis\_id}\in\{1,\ldots,d\}$ and
$\texttt{vector\_id}\in\{0,\ldots,d-1\}$, hence $d^2$ candidates.  It refuses
to materialize more than 70,000 states.  Every analytic stabilizer generator
is checked against the dense statevector for small registers, as are
commutation, binary independence, and the unique mixer ground state.

Family 1 is pinned by three independent identities: its field index is
$\texttt{basis\_id}-1=0$; its Clifford phase and controlled-$Z$ exponents
vanish; and its analytic stabilizers are signed one-local $X$ operators.
Streaming the local MaxCut discovery backup finds that family-1 vector 0 and
the baseline have identical stored energy, decoded ratio, and angle fields in
all 677 available graph-depth scans.  This is an exact implementation
invariant, not a statistical comparison.

The corrected finite-field recurrence reduces a sum of tap contributions
modulo two at every step.  At $r=8$, the old power table ended in
$203,177,354$ rather than the corrected $171,77,154$.  Direct old/new catalog
comparison shows that only noncomputational basis identifiers
$1,2,63,64,141,142,179,180$ remain unchanged; 248 of 256 change.  The repair
is commit \texttt{3ae8dd1}.  Because the full historical $r=8$ nonproduct
scan predates this commit, it is excluded as described in
Sec.~\ref{sec:gf2-defect}.

\subsection{Angle optimization}

At $p=1$, all candidates use the same grid
\begin{equation}
 \beta\in[0,\pi]\quad(33\text{ points}),\qquad
 \gamma\in[-\pi,\pi)\quad(65\text{ points}).
\end{equation}

This interval is a declared bounded optimization domain; no \(2\pi\) spectral period of the noncommuting QRAO Hamiltonian is assumed.

The three best grid points are refined with L-BFGS-B \cite{ByrdEtAl1995}, using
$\texttt{ftol}=10^{-12}$, $\texttt{gtol}=10^{-5}$,
$\texttt{maxfun}=15000$, and $\texttt{maxiter}=80$.  For $p>1$, the optimized
depth-$(p-1)$ vector is used as a parameter-continuation seed while a
$9\times13$ grid initializes the appended layer, followed by the same local
solver.  This continuation is unrelated to a problem-specific
classical-solution warm start: no approximate cut initializes the quantum
state.  Objective evaluation is noiseless double-precision statevector
simulation.

A seven-budget sweep on 30 already selected $\Xi_{123}$ states changes only
\texttt{maxiter}.  At $p=1$ all budgets from 10 through 1000 tie in decoded
ratio.  At $p=3$, budget 10 loses on 20 of 30 graphs relative to budget 80,
while budgets 40--1000 tie.  At $p=5$, budgets 10, 20, and 40 lose on 23, 13,
and one graphs, respectively, and budgets 80--1000 tie.  This validates the
chosen local budget for the tested winners but cannot rule out a different
winner under a different global state/angle search.

The campaign uses finite-difference gradients, so deeper optimizations are
sensitive to last-bit BLAS differences between machines.  Rows at $p\ge2$ are
not merged across compute environments.  The same-machine implementation was
checked for byte identity across the evaluation-cache optimization; this is a
software regression check, not a statistical result.

\begin{table*}[t]
  \centering
  \caption{Exhaustive decoded family means.  Every family contributes all
  $2^r$ labels on each graph.  The margin is a graph-paired
  product-minus-strongest-nonproduct difference; the interval is the
  descriptive graph-bootstrap interval for that selected comparator.}
  \label{tab:family-ranks}
  \small
\begin{tabular}{rrrrrrrr}
\toprule
$r$ & $p$ & graphs & product mean & rank & best nonproduct & margin & paired 95\% CI \\
\midrule
4 & 1 & 58 & 0.8024 & 1 & basis 5: 0.7791 & +0.0233 & [+0.0052, +0.0416] \\
5 & 1 & 136 & 0.7883 & 1 & basis 10: 0.7481 & +0.0401 & [+0.0307, +0.0497] \\
6 & 1 & 218 & 0.7842 & 1 & basis 2: 0.7401 & +0.0441 & [+0.0381, +0.0502] \\
7 & 1 & 56 & 0.7729 & 1 & basis 3: 0.7018 & +0.0711 & [+0.0602, +0.0820] \\
4 & 2 & 58 & 0.8164 & 1 & basis 5: 0.8010 & +0.0155 & [-0.0047, +0.0356] \\
5 & 2 & 136 & 0.8179 & 1 & basis 10: 0.7830 & +0.0349 & [+0.0267, +0.0431] \\
6 & 2 & 218 & 0.8003 & 1 & basis 2: 0.7493 & +0.0511 & [+0.0460, +0.0566] \\
7 & 2 & 56 & 0.7983 & 1 & basis 66: 0.7142 & +0.0841 & [+0.0767, +0.0912] \\
4 & 3 & 58 & 0.8359 & 1 & basis 5: 0.8225 & +0.0134 & [-0.0030, +0.0301] \\
5 & 3 & 136 & 0.8367 & 1 & basis 10: 0.7972 & +0.0395 & [+0.0322, +0.0465] \\
6 & 3 & 218 & 0.8071 & 1 & basis 2: 0.7570 & +0.0501 & [+0.0454, +0.0547] \\
7 & 3 & 56 & 0.8080 & 1 & basis 2: 0.7235 & +0.0845 & [+0.0767, +0.0925] \\
4 & 4 & 58 & 0.8485 & 2 & basis 5: 0.8535 & -0.0051 & [-0.0231, +0.0132] \\
5 & 4 & 136 & 0.8484 & 1 & basis 10: 0.8141 & +0.0343 & [+0.0272, +0.0414] \\
6 & 4 & 218 & 0.8182 & 1 & basis 2: 0.7643 & +0.0539 & [+0.0491, +0.0589] \\
7 & 4 & 56 & 0.8146 & 1 & basis 66: 0.7303 & +0.0843 & [+0.0759, +0.0929] \\
4 & 5 & 58 & 0.8552 & 2 & basis 5: 0.8653 & -0.0101 & [-0.0276, +0.0079] \\
5 & 5 & 136 & 0.8558 & 1 & basis 10: 0.8275 & +0.0283 & [+0.0212, +0.0356] \\
6 & 5 & 218 & 0.8237 & 1 & basis 2: 0.7722 & +0.0515 & [+0.0466, +0.0562] \\
7 & 5 & 56 & 0.8217 & 1 & basis 2: 0.7356 & +0.0861 & [+0.0780, +0.0942] \\
\bottomrule
\end{tabular}

\end{table*}

\begin{table*}[t]
 \centering
 \caption{Legacy-arm accounting behind the physical union column in
 Table~\ref{tab:resource-counts}.  The executed union is deduplicated, so its
 size is not the sum of arm sizes.  $R_E$ is retained here for provenance only.}
 \label{tab:resource-legacy}
 \small
\begin{tabular}{rrrrr}
\toprule
$r$ & legacy $R_E$ & displayed-arm union & physical executed union & legacy-only increment \\
\midrule
6 & 100 & 112 & 112 & 0 \\
7 & 116 & 180 & 184 & 4 \\
8 & 124 & 298 & 302 & 4 \\
9 & 140 & 410 & 414 & 4 \\
10 & 156 & 524 & 528 & 4 \\
11 & 172 & 637 & 641 & 4 \\
12 & 188 & 738 & 742 & 4 \\
\bottomrule
\end{tabular}
\par\smallskip
\footnotesize The displayed-arm union deduplicates the structured MUB selector $\Xi_{123}$, the $r^2$-budget Rand-SEMI-QAOA arm, and the capped matched-cardinality Rand-SEMI-QAOA arm. The physical union additionally includes legacy $R_E$ states; overlaps are counted once. A baseline run is not included in either union.

\end{table*}

\subsection{Graph-first aggregation and uncertainty}

State rows are never treated as independent replicates.  For a family mean,
we average its $d$ states within each graph.  For a candidate arm, we take its
maximum within each graph.  Only after this reduction do we compute a mean,
standard error, or bootstrap interval over graphs.  Every paired arm
difference is joined on $(n,p,r,\texttt{graph\_id})$ before resampling.
Percentile intervals use 10,000 graph resamples \cite{EfronTibshirani1993} and
deterministic seeds derived from the table key.  The regeneration script emits the seeds implicitly
through a SHA-256 key and records all source hashes.

The family runner-up in Table~\ref{tab:family-ranks} is chosen from the same
data used to report its margin.  Thus the intervals are descriptive paired
intervals for the selected runner-up, not selection-adjusted simultaneous
confidence intervals over all families.  Small cells are interpreted
accordingly.

\subsection{Candidate-set completeness and overlaps}

The four scale arms are evaluated through their distinct union: a state shared
by several arms is optimized once and credited to each.  Per-graph arm maxima
are emitted only after every expected union state for that graph and depth is
present.  This avoids downward bias from an unfinished task.  Selector sizes
are regenerated from the candidate code for every $r$ and cross-checked
against the stored union/unique-arm accounting.

The primary scale table contains exactly five arm records per complete scan:
the structured cross-family MUB selector $\Xi_{123}$ (internal key
\texttt{xi123}), a legacy ellipse $R_E$, the $r^2$-budget
\RandSEMIQAOA{} arm (\texttt{rand\_r2}), the matched-budget
\RandSEMIQAOA{} arm (\texttt{rand\_budget}), and the baseline.  The two random
product pools use separately seeded frozen subsets and are not generally
nested budget prefixes.  All three contributing campaigns use selector seed
12.  The implementation derives deterministic arm- and $r$-specific seeds
from a SHA-256 namespace, so membership depends only on the schedule, $r$, and
that seed and is reused across graph sizes, graph identifiers, and depths.  At
$r=6,7$ the capped matched-cardinality schedule saturates the complete product family.
The paper omits $R_E$ from main
figures because it contains a historical basis-3 branch with unresolved
$r=8$ provenance and because it is not needed to test the product-family
claim.  Its rows remain in the source data and completeness gate, and its
states remain included in the physical union counts reported in
Table~\ref{tab:resource-counts}.

\begin{table}[t]
 \centering
 \caption{Decoded best-of-set ratio on the complete $n=18,20,22$ cells,
 pooled over 2,400 graphs per depth.  This numerical companion to
 Fig.~\ref{fig:shallow-performance}b is placed here to avoid duplicating the main
 figure with a full table.}
 \label{tab:depth-summary}
 \small
\begin{tabular}{rrrrr}
\toprule
$p$ & baseline & \shortstack{Rand-SEMI\\$K=r^2$} & \shortstack{structured MUB\\selector $\Xi_{123}$} & \shortstack{Rand-SEMI\\matched $K$} \\
\midrule
1 & 0.7108 & 0.9319 & 0.9439 & 0.9421 \\
2 & 0.7516 & 0.9591 & 0.9760 & 0.9772 \\
3 & 0.7737 & 0.9738 & 0.9880 & 0.9893 \\
4 & 0.7810 & 0.9795 & 0.9916 & 0.9927 \\
5 & 0.7880 & 0.9829 & 0.9934 & 0.9942 \\
\bottomrule
\end{tabular}

\end{table}

\subsection{Runtime censoring}

The scale snapshot was produced while tasks were still completing.  Shorter
registers finish first.  This creates a known composition bias at large $n$:
the intended $n=30$ cohort has register histogram
$\{r10:16,r11:169,r12:15\}$, whereas the analysis freeze contains
$\{16,163,0\}$; the intended $n=34$ histogram is
$\{r12:123,r13:77\}$, whereas the analysis freeze contains $\{65,0\}$.  No model-based
correction is attempted.  Figure~\ref{fig:shallow-performance}a shows one continuous
trajectory, but changes to an open-marker, dashed-line encoding in the
censored tail so that visual continuity is not mistaken for an unbiased
cohort estimate.

\subsection{Energy/decoder diagnostics}

Within a graph--depth scan, energy--ratio alignment is computed over all stored
candidate states.  The regenerated analysis applies the same exact,
$r$-specific physical-union gate as the primary scale reduction.  It excludes
17 incomplete scans from the 12,449-row source snapshot and retains 12,432
complete scans.  Their 49,728 primary-arm maxima reproduce the main reduction
bit-for-bit.  The scan-weighted Spearman mean and median are $0.528499$ and
$0.531283$.  The energy argmax is also a decoded-ratio maximum, including
ties, in $0.222732$ of scans; the mean decoded loss over all retained scans is
$0.069907$.
Equal weighting over the 17 retained $(n,p)$ cells gives a Spearman mean of
$0.535951$, win rate $0.231103$, and loss $0.069387$.

The physical union includes the legacy $R_E$ contribution, so this diagnostic
is descriptive of the executed campaign rather than a comparison restricted
to the displayed arms.  Its $r=8$ legacy component is not used as corrected
nonproduct-family evidence.  The retained compact inputs are exact per-scan
sufficient statistics and per-arm maxima; raw candidate-state rows for this
scale analysis are not retained in the compact archive.

The second decoder diagnostic uses the analytic expectation
\begin{equation}
 C_{\mathrm{magic}}=\frac49W+\frac19\langle H\rangle,
\end{equation}
where $W$ is the constant problem term under that convention.  It tests score
ordering under a magic-basis expectation; it is not a finite-shot deployment
estimate.  We therefore say ``robust under the two tested score maps,'' not
``decoder independent.''

\begin{table*}[t]
 \centering
 \caption{Scope of the numerical evidence.  ``Rows retained'' distinguishes retained raw state rows from summary-only
metadata.}
 \label{tab:experiment-matrix}
 \small
 \begin{tabularx}{\textwidth}{lYYYY}
 \toprule
 Study & Graphs and depth & Candidate statistic & Retained backing & Main use\\
 \midrule
 MaxCut exhaustive & $n=10$--18, $p=1$--5, 100/cell & all $4^r$ states & recovered
 authenticated 20,310,980-row merge & 20 validated graph-first family aggregations
 at $r=4$--7, $p=1$--5; historical $r=8$ nonproduct rows excluded\\
 No-QAOA & 72 graphs, nine $(n,r)$ cells, $p=0,1$ & all product labels for
 $r\le8$, 64 otherwise & 10,384-row per-state CSV & within-product lift and exact
 best-of-$K$\\
 Scale & $n=18$--34, $p=1$--5 & best of fixed candidate arms & 67,685
 best-per-arm rows, 13,537 complete scans & complete depth cells through
 $n=22$; descriptive $p=1$ scale snapshot\\
 Energy alignment & 12,432 complete graph--depth scans & within-scan state
 ordering after the primary union gate & regenerated CSV and JSON audit &
 scope separation between energy and decoded ratio\\
 \bottomrule
 \end{tabularx}
\end{table*}

\section{Proofs for the deterministic mechanism}
\label{app:det-proofs}

\subsection{Binary Pauli convention}

For $x,z\in V=\F_2^r$, let $P(x,z)$ denote a Hermitian representative of the
Pauli class with $X$ mask $x$ and $Z$ mask $z$.  Phases are chosen so that the
spread subgroup
$L_j=\{P(x,M_jx):x\in V\}$ consists of commuting Hermitian operators.  Its
eigenstate labels are fixed by
\begin{equation}
 P(x,M_jx)\ket{\psi_{j,v}}=(-1)^{v\cdot x}\ket{\psi_{j,v}}.
 \label{eq:spread-eigenvalue}
\end{equation}
Distinct Paulis are orthogonal under the normalized trace:
$\tau_r(PP')=0$ unless $P=P'$ up to the chosen Hermitian phase.

\begin{proof}[Proof of Lemma~\ref{lem:stabilizer-filter}]
For a Pauli $P(x,z)$, its expectation in a stabilizer state vanishes unless it
belongs to the state's maximal stabilizer subgroup.  Membership in $L_j$ is
equivalent to $z=M_jx$.  In that case Eq.~\eqref{eq:spread-eigenvalue} gives
the expectation $(-1)^{v\cdot x}$.  Applying this rule term by term to the
Pauli expansion of $O$ gives Eq.~\eqref{eq:stabilizer-filter}.  Averaging over
$v$ uses character orthogonality,
\begin{equation}
 2^{-r}\sum_{v\in V}(-1)^{v\cdot x}=\mathbf1_{x=0}.
\end{equation}
Only the identity coefficient remains, and that coefficient is $\tau_r(O)$.
\end{proof}

\begin{proof}[Proof of Lemma~\ref{lem:galois-randomization}]
If $M_j^{\mathsf T}\kappa=M_{j'}^{\mathsf T}\kappa$, symmetry of the spread
matrices gives $(M_j-M_{j'})\kappa=0$.  For $j\ne j'$, the spread property
makes $M_j-M_{j'}$ nonsingular, contradicting $\kappa\ne0$.  The map is
injective on a finite set of size $2^r$ and therefore bijective.
\end{proof}

\subsection{Matched-mixer conjugation}

We first record the sign in the physical convention.  If Hermitian Paulis
$g,T$ anticommute, then
\begin{align}
 e^{-\ii\beta g}T e^{\ii\beta g}
 &=(\cos\beta I-\ii\sin\beta g)T
   (\cos\beta I+\ii\sin\beta g)\nonumber\\
 &=\cos(2\beta)T-\ii\sin(2\beta)gT.
 \label{eq:single-generator-conjugation}
\end{align}
For commuting $g,T$, the conjugation is $T$.  Since the $g_{j,v,q}$ commute
with one another, selecting the sine branch on a subset
$\kappa\subseteq S_T$ gives
\begin{equation}
 e^{\ii\beta B_{j,v}}T e^{-\ii\beta B_{j,v}}
 =\sum_{\kappa\subseteq S_T}
 c_\beta^{w_j(T)-|\kappa|}(-\ii s_\beta)^{|\kappa|}G_\kappa T,
 \label{eq:full-mixer-conjugation}
\end{equation}
because $B_{j,v}=-\sum_qg_{j,v,q}$.

\begin{proof}[Proof of Theorem~\ref{thm:exact-branches}]
Expand the observable $H=\sum_Th_TT$ and insert
Eq.~\eqref{eq:full-mixer-conjugation}.  Since $G_\kappa^2=I$,
\begin{align}
 e^{\ii\gamma H}G_\kappa T e^{-\ii\gamma H}
 &=G_\kappa e^{\ii\gamma G_\kappa HG_\kappa}
 T e^{-\ii\gamma H}\nonumber\\
 &=G_\kappa e^{\ii\gamma H^{(\kappa)}}T e^{-\ii\gamma H}.
\end{align}
The signed product $G_\kappa$ is a stabilizer of
$\ket{\psi_{j,v}}$ with eigenvalue $+1$.  It may therefore be removed from
the expectation.  Averaging the remaining operator over $v$ and applying
Eq.~\eqref{eq:family-trace} produces the normalized trace
$L_{T,\kappa}$.  For $\kappa=0$, cyclicity gives
\begin{equation}
 L_{T,0}=\tau_r(e^{\ii\gamma H}Te^{-\ii\gamma H})
 =\tau_r(T)=0
\end{equation}
because every root term is nonidentity.  This proves
Eq.~\eqref{eq:exact-branch-formula}.  At $\beta=k\pi/2$ every nonempty branch
contains $s_\beta=0$, proving the resonance statement.
\end{proof}

\subsection{Rank defect and locality of the implemented packing}

\begin{proof}[Proof of Proposition~\ref{prop:actual-packing}]
Let $G$ be a connected cubic graph on $n$ vertices.  The encoder partitions
each of its $\chi\le4$ color classes into chunks of sizes one, two, or three.
Let $r$ be the total number of chunks and $n_1$ the number of singleton
chunks.  Define the linear map
\begin{equation}
 L:\F_2^n\longrightarrow\F_2^{2r}
\end{equation}
that sends the coordinate vector of a graph vertex to the binary symplectic
label of its assigned one-qubit Pauli.

The local contribution to $\rank L$ is one for a singleton and two for a
chunk of size two or three: any two distinct vectors among the one-qubit
labels $X,Y,Z$ are linearly independent, and $X+Y+Z=0$.  Supports from
different relaxed qubits are disjoint.  Hence
\begin{equation}
 \rank L=2r-n_1.
 \label{eq:L-rank}
\end{equation}
At most one remainder chunk in each color class is a singleton, so
$n_1\le\chi\le4$.

The binary incidence vectors $e_i+e_j$ of a connected graph span the
even-parity subspace
\begin{equation}
 W=\{a\in\F_2^n:\textstyle\sum_i a_i=0\},
 \qquad\dim W=n-1.
\end{equation}
The symplectic label of edge term $P_iP_j$ is $L(e_i+e_j)$, so the row span of
$[Z\ X]$ is $L(W)$.  Restricting a linear map to a codimension-one subspace
reduces its rank by at most one; therefore
\begin{equation}
 \rank\Phi=\dim L(W)\ge\rank L-1=2r-n_1-1.
\end{equation}
It follows that
\begin{equation}
 d_r=2r-\rank\Phi\le n_1+1\le5.
 \label{eq:rank-defect-five}
\end{equation}

For a color class of size $m$, the number of chunks is $\lceil m/3\rceil$.
Summing over colors gives
\begin{equation}
 \frac n3\le r\le\frac n3+\frac{2\chi}{3}\le\frac n3+\frac83.
 \label{eq:r-versus-n}
\end{equation}
The vertex-to-Pauli map is injective: within a chunk the axes are distinct,
and different chunks use different relaxed qubits.  Distinct graph edges
therefore produce distinct two-qubit Pauli strings.  Since a cubic graph has
$M_r=3n/2$ edges, $M_r=\Theta(r)$ and $M_r/r\to9/2$.

Each chunk contains at most three degree-three vertices, so at most nine
Hamiltonian edges touch a relaxed qubit.  An edge term acts on two relaxed
qubits; excluding itself, it can overlap at most eight terms through each.
Thus the term-overlap degree is at most 16.

Finally, an edge term has zero $X$ label exactly when both endpoint axes are
$Z$.  There is at most one $Z$-assigned vertex in a chunk, hence at most $r$
such vertices.  The number of $ZZ$ edges is at most $3r/2$, so
\begin{equation}
 \frac{M_r^*}{M_r}\ge1-\frac rn\longrightarrow\frac23.
 \label{eq:active-density}
\end{equation}
For the product family $M_0=0$, a MaxCut edge has zero syndrome exactly when
both endpoint axes are $X$.  The same argument gives the final density claim.
\end{proof}

\subsection{Positive-distance counting}

Let $B_M(\delta M)=\{q\in\F_2^M:\wt(q)\le\delta M\}$.

\begin{proof}[Proof of Theorem~\ref{thm:code-distance}]
Fix $\kappa\ne0$.  By Lemma~\ref{lem:galois-randomization},
$y_j=M_j^{\mathsf T}\kappa$ visits every $y\in V$ exactly once as $j$ varies.
Using symmetry of $M_j$,
\begin{equation}
 \mathcal A_j\kappa=Z\kappa+XM_j\kappa=\Phi(\kappa,y_j).
\end{equation}
Thus every bad pair $(j,\kappa)$ maps to a preimage under $\Phi$ of a vector
in $B_{M_r}(\delta M_r)$.  Every nonempty fiber of $\Phi$ has size
$2^{d_r}\le2^{d_0}$, so
\begin{align}
 &\#\{(j,\kappa):\kappa\ne0,
   \wt(\mathcal A_j\kappa)\le\delta M_r\}\nonumber\\
 &\qquad\le2^{d_0}\sum_{q=0}^{\lfloor\delta M_r\rfloor}
 \binom{M_r}{q}
 \le2^{d_0+M_rH_2(\delta)}.
 \label{eq:bad-pair-count}
\end{align}
The entropy condition makes this at most
$2^{(1-\eta)r+o(r)}$.  The number of bad families is no larger than the number
of bad pairs.  Division by $2^r$ gives an exceptional relative fraction
$2^{-\eta r+o(r)}$.

For the implemented cubic packing, Proposition~\ref{prop:actual-packing}
allows $d_0=5$ and $M_r/r\to9/2$.  Any fixed sufficiently small $\delta>0$
with $(9/2)H_2(\delta)<1$ satisfies the condition for large $r$.
\end{proof}

\subsection{Connected two-contour series}

For fixed $j,\kappa$, write
$\sigma_A=(-1)^{(\mathcal A_j\kappa)_A}$ and define
\begin{align}
 Z_\kappa(\gamma)
 &=\tau_r(e^{\ii\gamma H^{(\kappa)}}e^{-\ii\gamma H}),\nonumber\\
 Z_{T,\kappa}(\gamma,z)
 &=\tau_r(e^{\ii\gamma H^{(\kappa)}}e^{zT}e^{-\ii\gamma H}).
\end{align}
On a disk where the connected logarithm converges,
\begin{equation}
 L_{T,\kappa}=Z_\kappa R_{T,\kappa},
 \qquad
 R_{T,\kappa}=\left.\partial_z\log Z_{T,\kappa}\right|_{z=0}.
 \label{eq:det-root-factorization}
\end{equation}

We spell out the connected-series estimate used here.  A two-contour word is
an ordered list of Hamiltonian-term occurrences, each assigned to the forward
or backward exponential.  Its occurrence graph connects two occurrences when
their physical Pauli supports overlap.  Normalized trace and exponential
shuffle coefficients factor across disconnected components.  The exponential
formula therefore expresses $\log Z_\kappa$ as the sum of connected occurrence
graphs; differentiating a source logarithm selects connected graphs containing
the root $T$.

At total order $m$, choose one root occurrence in at most $M_r$ ways.  A
connected occurrence graph has a spanning tree; Cayley's bound gives at most
$m^{m-2}$ labeled trees, and each child term has at most $D$ choices after its
parent.  The two contours contribute
\begin{equation}
 \sum_{a+b=m}\frac{|\gamma|^m}{a!b!}
 =\frac{(2|\gamma|)^m}{m!}.
\end{equation}
Since every normalized Pauli trace has modulus at most one, the connected
order-$m$ contribution is bounded by
\begin{equation}
 M_r\frac{(2h_*|\gamma|)^m}{m!}m^{m-2}D^{m-1}.
 \label{eq:det-tree-bound}
\end{equation}
Using $m^{m-2}/m!\le e^m/m^2$, the series converges absolutely and uniformly
whenever $2eh_*D|\gamma|<1$.  For the rooted series, the root $T$ is fixed and
the factor $M_r$ is absent.

Let $K=H^{(\kappa)}-H=-2\sum_{A:\sigma_A=-1}h_AA$.  Direct second-order
expansion, trace cyclicity, and Pauli orthogonality give
\begin{equation}
 Z_\kappa(\gamma)=1-\frac{\gamma^2}{2}\tau_r(K^2)
 +O(M_r|\gamma|^3)
 =1-2Q_j(\kappa)\gamma^2+O(M_r|\gamma|^3).
\end{equation}
The uniform connected bound upgrades this to
\begin{equation}
 \log Z_\kappa=-2Q_j(\kappa)\gamma^2+\mathcal R_\kappa,
 \qquad |\mathcal R_\kappa|\le C_0M_r|\gamma|^3.
 \label{eq:det-log-expansion}
\end{equation}
At $\gamma=0$, $R_{T,\kappa}=\tau_r(T)=0$.  Every nonzero rooted connected
word contains at least one Hamiltonian occurrence, so the rooted version of
Eq.~\eqref{eq:det-tree-bound} yields
\begin{equation}
 |R_{T,\kappa}(\gamma)|\le C_1|\gamma|.
 \label{eq:det-root-bound}
\end{equation}

\begin{proof}[Proof of Theorem~\ref{thm:det-dephasing}]
If $Q_j(\kappa)\ge q_*r$ and $M_r\le C_Mr$, choose a uniform
$\gamma_*>0$ such that
$C_0C_M|\gamma|^3\le q_*\gamma^2$ for
$|\gamma|\le\gamma_*$.  Equation~\eqref{eq:det-log-expansion} gives
$|Z_\kappa|\le e^{-q_*r\gamma^2}$.  Combining this with
Eqs.~\eqref{eq:det-root-factorization} and \eqref{eq:det-root-bound} proves
Eq.~\eqref{eq:det-dephasing}.

For equal-magnitude coefficients $|h_A|=h$, positive distance gives
$Q_j(\kappa)=h^2\wt(\mathcal A_j\kappa)\ge h^2\delta M_r=\Omega(r)$
simultaneously for every nonzero branch.
\end{proof}

\subsection{Product-family rooted locality}

The product family has an additional structural property that does not require
a connected logarithm.  For a root Pauli $T$, its mixer-rotated observable
$O_T(\beta)=e^{\ii\beta B_{0,v}}Te^{-\ii\beta B_{0,v}}$ remains supported on
$\supp T$ because $B_{0,v}$ is a sum of one-local $X$ operators.  Expand
\begin{equation}
 e^{\ii\gamma H}O_T(\beta)e^{-\ii\gamma H}
 =\sum_{m\ge0}\frac{(\ii\gamma)^m}{m!}
 \operatorname{ad}_H^m(O_T).
 \label{eq:nested-commutator}
\end{equation}
In a nested commutator, a newly selected Hamiltonian term contributes only if
it overlaps the support union already generated by $T$ and the previous
terms; otherwise it commutes and that word vanishes.  Every nonzero word is
therefore a physical occurrence cluster connected to $T$.  The raw number of
ordered words need not be exponential in $m$.  Rather, the rooted spanning-tree
bound used in Eq.~\eqref{eq:det-tree-bound}, together with the $1/m!$ Taylor
weight, bounds the factorial-weighted sum of order-$m$ rooted occurrence words
by $O(C^m)$ uniformly in $r$.  This exact support preservation is the
deterministic reason product-family Taylor coefficients do not acquire a
volume factor after the $M_r^{-1}$ root average.

The argument says nothing about the sign of connected cubic or higher-order
coefficients.  In particular, Pauli closure of a proposed word is necessary
but not sufficient for a nonzero signed contribution after contour and branch
summation.  No deterministic cubic family ordering is claimed.

\section{Proofs for the anisotropic Gaussian mechanism}
\label{app:gaussian-proofs}

\subsection{Response and isotropic no-go}

Fix the QRAO Pauli frame $\{A_e:e\in E_r\}$ and let
\begin{equation}
 H_G=h\sum_{e\in E_r}g_eA_e,
 \qquad g_e\stackrel{\mathrm{iid}}{\sim}N(0,1).
\end{equation}
For common angles $(\gamma,\beta)$, define the annealed label-averaged
one-layer energy by
\begin{align}
 \mathcal E^{G}_{j,r}(\gamma,\beta)
 &:=\frac{1}{M_r2^r}\sum_{v\in V}\E_g\Bigl[
 \bra{\psi_{j,v}}e^{\ii\gamma H_G}e^{\ii\beta B_{j,v}}\nonumber\\[-2pt]
 &\hspace{4.2em}{}\times H_G\,e^{-\ii\beta B_{j,v}}
 e^{-\ii\gamma H_G}\ket{\psi_{j,v}}\Bigr],
 \label{eq:gaussian-annealed-energy}
\end{align}
and its response by
\begin{equation}
 \mathcal F_{j,r}(\gamma,\beta)
 :=\mathcal E^{G}_{j,r}(\gamma,\beta)
   -\mathcal E^{G}_{j,r}(0,\beta).
 \label{eq:gaussian-response-definition}
\end{equation}
This is a signed relaxed-energy change normalized by the number of edge terms.
It is not a quenched statement, a statewise angle optimization, or a decoded
score.

Suppose instead that $G$ were a centered isotropic Gaussian on the real
Hilbert space of traceless Hermitian operators.  If a Clifford $W_j$ maps the reference family and its matched mixers to
family $j$, let $\mathcal R_{j,v}(G;\theta)$ denote the corresponding
label-response functional at angle tuple $\theta$.  Circuit conjugation gives
\begin{equation}
 \mathcal R_{j,v}(G;\theta)
 =\mathcal R_{0,v}(W_j^\dagger GW_j;\theta).
\end{equation}
Isotropy implies $W_j^\dagger GW_j\stackrel{d}=G$, so all finite-dimensional
distributions of the complete response process are the same for every $j$.
No covariant statistic of that process can select the product family.  The
sparse covariance
\begin{equation}
 \Sigma_{\mathrm{QRAO}}(X)=h^2\sum_e A_e\Tr(A_eX)
\end{equation}
is therefore essential: it has rank at most $M_r=O(r)$ in a traceless
operator space of dimension $4^r-1$.

For completeness, the static process
$X_{j,v}=\langle\psi_{j,v}|H_G|\psi_{j,v}\rangle$ has exact covariance
\begin{equation}
 \E[X_{j,v}X_{j,v'}]
 =h^2\sum_{e:a_j(A_e)=0}(-1)^{(v+v')\cdot x_e}.
 \label{eq:static-covariance}
\end{equation}
If $x_e\ne0$, the spread property lets that edge satisfy
$a_j(A_e)=0$ for at most one Galois family.  Thus at most $M_r$ families have
any active static edge, another expression of anisotropy.  We do not use this
static fact to infer a decoded advantage.

\subsection{Exact-support two-contour polymers}

For a mixer branch $\kappa$, put
\begin{equation}
 \sigma_e(\kappa)=(-1)^{(\mathcal A_j\kappa)_e},
 \qquad H_\kappa=h\sum_e\sigma_e(\kappa)g_eA_e.
\end{equation}
Define
\begin{align}
 Z_\kappa(\gamma)
 &=\E_g\tau_r(e^{\ii\gamma H_\kappa}e^{-\ii\gamma H_G}),\nonumber\\
 Z_{T,\kappa}(\gamma,z)
 &=\E_g\tau_r(e^{\ii\gamma H_\kappa}e^{zg_tT}
 e^{-\ii\gamma H_G}),\nonumber\\
 X_{T,\kappa}(\gamma)
 &=\left.\partial_zZ_{T,\kappa}(\gamma,z)\right|_{z=0},
 \label{eq:gaussian-vacuum-source}
\end{align}
where $T=A_t$.  On the convergence disk,
\begin{equation}
 X_{T,\kappa}=Z_\kappa R_{T,\kappa},
 \qquad R_{T,\kappa}=\left.\partial_z\log Z_{T,\kappa}\right|_{z=0}.
 \label{eq:gaussian-root-factorization}
\end{equation}

Expand both contour exponentials.  For a nonempty label set
$C\subseteq E_r$, let $W_\kappa(C;\gamma)$ be the sum of all annealed words
whose set of distinct Gaussian labels is exactly $C$, including all
assignments of occurrences to contours and all word orders.  A label must
occur an even number of times to survive Gaussian averaging.

If the induced overlap graph on $C$ has connected components
$C_1,\ldots,C_m$, the Gaussian moments factor because the label sets are
disjoint.  The Pauli words commute and the normalized trace factors because
the components have disjoint physical qubit support.  Finally, the number of
shuffles of component occurrences cancels the multinomial decomposition of
the exponential factorials.  Hence
\begin{equation}
 W_\kappa(C;\gamma)=\prod_{i=1}^mW_\kappa(C_i;\gamma).
 \label{eq:exact-support-factorization}
\end{equation}
Summing over finite label sets and decomposing them into components gives the
exact hard-core polymer gas
\begin{equation}
 Z_\kappa(\gamma)=
 \sum_{\mathcal P\ \mathrm{compatible}}
 \prod_{C\in\mathcal P}W_\kappa(C;\gamma),
 \label{eq:hard-core-gas}
\end{equation}
where polymers are nonempty connected edge sets and compatibility means graph
distance at least two.  Equation~\eqref{eq:hard-core-gas} retains every
noncommuting ordering within a connected set and every Gaussian Wick pairing.

\subsection{Activity and Koteck\'y--Preiss bounds}

Set $u_{\rm abs}(\gamma)=e^{2h^2|\gamma|^2}-1$, abbreviated below by
$u_{\rm abs}$.  If label $e$ occurs $2m$ times, the sum of absolute scalar
weights is bounded by
\begin{align}
 \sum_{m\ge1}\frac{(2|h\gamma|)^{2m}}{(2m)!}\E g_e^{2m}
 &=\sum_{m\ge1}\frac{(2h^2|\gamma|^2)^m}{m!}
 =u_{\rm abs}(\gamma),
\end{align}
where $\E g^{2m}=(2m-1)!!$.  Every Pauli trace has modulus at most one, so
\begin{equation}
 |W_\kappa(C;\gamma)|\le u_{\rm abs}(\gamma)^{|C|}.
 \label{eq:activity-bound-proof}
\end{equation}

Let $D$ be one plus the maximum overlap degree.  The number of connected edge
sets of size $s$ containing a fixed edge is at most $(eD)^{s-1}$.  A polymer
incompatible with $C_0$ contains a marked edge at graph distance at most one
from $C_0$, with at most $D|C_0|$ choices.  Therefore
\begin{align}
 \sum_{C\not\sim C_0}|W_\kappa(C)|e^{|C|}
 &\le D|C_0|\sum_{s\ge1}(eD)^{s-1}(u_{\rm abs}e)^s\nonumber\\
 &=D|C_0|\frac{eu_{\rm abs}}{1-e^2Du_{\rm abs}}.
 \label{eq:kp-sum}
\end{align}
For
\begin{equation}
 u_{\rm abs}\le\rho_D:=\frac1{16e^2D^2},
\end{equation}
the final factor is at most $|C_0|$.  The Koteck\'y--Preiss criterion
\cite{KoteckyPreiss1986} then gives absolute, uniform convergence of the gas,
its connected logarithm, and source derivatives.

\subsection{Singleton resummation and fourth-order remainder}

For a singleton label $e$, both contour factors are functions of the same
Pauli and commute.  The Gaussian characteristic function yields
\begin{align}
 1+W_\kappa(\{e\};\gamma)
 &=\E\tau_r\!\left[e^{\ii\gamma h(\sigma_e-1)g_eA_e}\right]\nonumber\\
 &=\exp\!\left[-\tfrac12h^2\gamma^2(1-\sigma_e)^2\right]
 =\begin{cases}
 1,&\sigma_e=+1,\\
 e^{-2h^2\gamma^2},&\sigma_e=-1.
 \end{cases}
 \label{eq:singleton-exact}
\end{align}
Thus the exact singleton part of $\log Z_\kappa$ is
$-2h^2N_-(\kappa)\gamma^2$.

We bound everything else.  For fixed $\gamma\ne0$, write
$W(C;\gamma)=u_{\rm abs}(\gamma)^{|C|}\widehat W(C;\gamma)$ with
$|\widehat W(C;\gamma)|\le1$ (the case $\gamma=0$ follows by continuity),
and let $Z_\kappa^{(s)}$ denote the same gas after replacing the scalar
$u_{\rm abs}(\gamma)$ by a complex variable $s$.  The KP disk contains
$|s|\le\rho_D$.  Define
\begin{equation}
 Q(s)=\log Z_\kappa^{(s)}
 -\sum_e\log(1+s\widehat W(\{e\};\gamma)).
\end{equation}
Its constant and linear Taylor coefficients vanish.  Uniform KP bounds give
$\sup_{|s|\le\rho_D}|Q(s)|\le A_DM_r$ for a constant depending only on $D$.
Cauchy's estimate then gives, for $|s|\le\rho_D/2$,
\begin{equation}
 |Q(s)|\le K_DM_r|s|^2.
\end{equation}
Returning to $s=u_{\rm abs}(\gamma)$ proves
\begin{equation}
 \log Z_\kappa=-2h^2N_-(\kappa)\gamma^2+\mathcal R_\kappa,
 \qquad |\mathcal R_\kappa|\le K_DM_ru_{\rm abs}(\gamma)^2.
 \label{eq:fourth-order-remainder}
\end{equation}
If $2h^2|\gamma|^2\le1$, then
$u_{\rm abs}\le4h^2|\gamma|^2$, so the remainder is bounded
by $C_4M_r|\gamma|^4$.  Distinct-label Gaussian contractions cannot occur at
quadratic order; every connected multilabel word contains at least four
occurrences.  This is the structural origin of the fourth-order start.

For a rooted source polymer, label $t$ occurs an odd number $2m+1$ of times,
while every other label occurs a positive even number.  Its root scalar sum is
\begin{equation}
 \sum_{m\ge0}\frac{(2|h\gamma|)^{2m+1}(2m+1)!!}{(2m+1)!}
 =2|h\gamma|e^{2h^2|\gamma|^2}=:v_{\rm abs}(\gamma).
\end{equation}
Hence a rooted activity supported on $C\ni t$ is bounded by
$v_{\rm abs}(\gamma)u_{\rm abs}(\gamma)^{|C|-1}$.  The rooted version of
Eq.~\eqref{eq:kp-sum} gives
\begin{equation}
 |R_{T,\kappa}(\gamma)|\le C_R|h\gamma|.
\end{equation}
Combining this with Eqs.~\eqref{eq:gaussian-root-factorization} and
\eqref{eq:fourth-order-remainder} proves
Eq.~\eqref{eq:gaussian-root-bound} and Theorem~\ref{thm:polymer-control}.

\subsection{Simultaneous activity of low-weight branches}

Call an edge active when $x_e\ne0$, and let $M_r^*$ be the number of active
edges.  Assume eventually $M_r^*/M_r\ge p_*>0$.  Fix $\kappa\ne0$ and set
$Y_j=M_j^{\mathsf T}\kappa$.  Lemma~\ref{lem:galois-randomization} makes
$Y_j$ exactly uniform on $V$ for uniform $j$.  For an active edge,
\begin{equation}
 (\mathcal A_j\kappa)_e=\kappa\cdot z_e+Y_j\cdot x_e
\end{equation}
is balanced, so the expected number $N_-^*$ of flipped active edges is
$M_r^*/2$.

View uniform $Y_j$ as $r$ independent fair bits.  Flipping coordinate $q$ can
change only indicators for terms with $(x_e)_q=1$.  All such terms act on
relaxed qubit $q$ and form a clique in the physical overlap graph, so there
are at most $D$ of them.  McDiarmid's inequality gives
\begin{equation}
 \Prob_j\!\left(\left|N_-^*-\frac{M_r^*}{2}\right|\ge t\right)
 \le2\exp\!\left(-\frac{2t^2}{rD^2}\right).
 \label{eq:mcdiarmid}
\end{equation}
Since $M_r=\Theta(r)$, choosing a fixed linear deviation shows
\begin{equation}
 \Prob_j(N_-^*<p_*M_r/4)\le e^{-c_Gr}.
\end{equation}
There are only
\begin{equation}
 \sum_{\ell=1}^{\lfloor\log r\rfloor}\binom r\ell
 =\exp(O((\log r)^2))
\end{equation}
nonzero branches of weight at most $\lfloor\log r\rfloor$.  A union bound
therefore makes $N_-(j,\kappa)\ge p_*M_r/4$ simultaneous for all such branches
outside a family set of relative size $e^{-\Omega(r)}$.

\begin{proof}[Proof of Theorem~\ref{thm:gaussian-dephasing}]
Let $\beta=b/r$ with $|b|\le b_+$.  For root syndrome size $w\le r$, the
total absolute mixer mass at branch size $\ell$ obeys
\begin{equation}
 \binom w\ell|\sin(2b/r)|^\ell
 \le\frac{(2b_+)^\ell}{\ell!}.
 \label{eq:mixer-factorial}
\end{equation}
For the simultaneous low-weight branches, the rooted bound becomes
\begin{equation}
 |X_{T,\kappa}|
 \le C_R|h\gamma|
 \exp\!\left[-\frac{p_*h^2}{2}M_r\gamma^2+C_4M_r\gamma^4\right].
\end{equation}
Summing Eq.~\eqref{eq:mixer-factorial} through
$\ell\le\lfloor\log r\rfloor$ costs at most $e^{2b_+}$.  For larger branches,
unitarity gives the trivial bound
$|X_{T,\kappa}|\le\E|g_t|$.  Their total contribution is bounded by
\begin{equation}
 \eta_r:=C_B\sum_{\ell>\lfloor\log r\rfloor}
 \frac{(2b_+)^\ell}{\ell!}=o(r^{-m})
 \label{eq:gaussian-tail}
\end{equation}
for every fixed $m$, where $C_B$ is independent of $r$.  Averaging over roots
then proves Eq.~\eqref{eq:gaussian-family-bound}.

Choose $\gamma_*$ so that the KP hypotheses hold and
$C_4\gamma_*^2\le p_*h^2/4$.  Then the first term is bounded by
$C|\gamma|e^{-aM_r\gamma^2}$.  Its supremum over $\gamma$ is
$(2eaM_r)^{-1/2}$ up to the constant, which vanishes; so does $\eta_r$.
This proves Eq.~\eqref{eq:gaussian-vanish}.  Proposition~\ref{prop:actual-packing}
supplies any eventual $p_*<2/3$ for the implemented encoder.
\end{proof}

\subsection{Corrected product-family slope}

The exact Gaussian branch expansion associated with
Eq.~\eqref{eq:gaussian-response-definition} is
\begin{equation}
 \mathcal F_{j,r}(\gamma,\beta)
 =\frac h{M_r}\sum_T
 \sum_{\varnothing\ne\kappa\subseteq S_T}
 c_\beta^{w_j(T)-|\kappa|}(-\ii s_\beta)^{|\kappa|}
 X_{T,\kappa}(\gamma).
 \label{eq:gaussian-branch-expansion}
\end{equation}
At $\gamma=0$, Gaussian integration or direct differentiation gives
\begin{equation}
 X'_{T,\kappa}(0)=\ii h(\sigma_T(\kappa)-1).
 \label{eq:root-slope}
\end{equation}
For $\kappa\subseteq S_T$,
$\sigma_T(\kappa)=(-1)^{|\kappa|}$, so only odd branches contribute.  Their
coefficient sum is
\begin{align}
 \sum_{|\kappa|\ \mathrm{odd}}
 c_\beta^{w-|\kappa|}(-\ii s_\beta)^{|\kappa|}
 &=\frac{(c_\beta-\ii s_\beta)^w-(c_\beta+\ii s_\beta)^w}{2}\nonumber\\
 &=-\ii\sin(2\beta w).
 \label{eq:odd-branch-sum}
\end{align}
Combining Eqs.~\eqref{eq:root-slope} and \eqref{eq:odd-branch-sum} yields
\begin{equation}
 \left.\partial_\gamma\mathcal F_{0,r}\right|_{\gamma=0}
 =-\frac{2h^2}{M_r}\sum_{e\in E_r}
 \sin(2\beta w_0(A_e)).
 \label{eq:correct-product-slope}
\end{equation}
The sum includes every edge.  Edges with $w_0=0$ contribute zero
automatically; $ZZ$ edges have $w_0=2$ and generally contribute.

It remains to justify the remainder uniformly in volume.  Fix a product-family
root $T=A_t$ and a branch $\kappa\subseteq S_T$, and let
\begin{equation}
 \mathcal D_\kappa=\{e:\sigma_e(\kappa)=-1\}.
\end{equation}
Every term in $\mathcal D_\kappa$ acts on a mixer qubit in $\kappa$ and hence
overlaps $T$.  Consequently $|\mathcal D_\kappa|\le D$.  Moreover
$|S_T|=w_0(T)\le w_*$, so the number of branches for a fixed root is bounded by
$2^{w_*}$ independently of $r$.

The exact-support gas has an additional cancellation in this setting.  If a
nonempty label set $C$ is disjoint from $\mathcal D_\kappa$, its signs all obey
$\sigma_e=1$.  Its exact-support coefficient is therefore the coefficient with
support exactly $C$ in
\begin{equation}
 \E_g\tau_r(e^{\ii\gamma H_C}e^{-\ii\gamma H_C})=1,
 \qquad H_C=h\sum_{e\in C}g_eA_e,
\end{equation}
and hence $W_\kappa(C;\gamma)=0$.  Thus every nonzero vacuum polymer is anchored
to the uniformly bounded set $\mathcal D_\kappa$.  Applying the rooted-cluster
version of Eq.~\eqref{eq:kp-sum}, with a marked label in that set, gives
$|\log Z_\kappa|\le C_D|\mathcal D_\kappa|$ on a smaller disk.  The source
logarithm is likewise uniform because every source polymer is connected to the
fixed root label $t$.  Hence that common complex disk carries uniform analytic
bounds on $Z_\kappa$, $R_{T,\kappa}$, and
$X_{T,\kappa}=Z_\kappa R_{T,\kappa}$.

The branch expansion contains only the uniformly bounded collection
$\kappa\subseteq S_T$, and its trigonometric coefficients have modulus at most
one.  After the $M_r^{-1}$ root average, $\mathcal F_{0,r}$ is therefore
uniformly bounded and analytic on that same disk.  Gaussian sign symmetry makes
it odd in $\gamma$.  Cauchy's estimate on a smaller disk, together with
Eq.~\eqref{eq:correct-product-slope}, now gives uniformly in $r$ and $\beta$
\begin{equation}
 \mathcal F_{0,r}(\gamma,\beta)
 =-2h^2s_{0,r}(\beta)\gamma+O(|\gamma|^3),
\end{equation}
which proves Eq.~\eqref{eq:product-linear-response}.

For the implemented cubic packing, $w_0(A_e)=0$ exactly on $XX$ edges.  There
is at most one $X$ vertex per chunk, so at most $3r/2$ edges are $XX$.  At
$\beta_0=\pi/8$, every remaining edge has $w_0=1$ or 2 and contributes at
least $1/\sqrt2$.  Therefore
\begin{equation}
 s_{0,r}(\pi/8)
 \ge\frac1{\sqrt2}\left(1-\frac rn\right)
 \longrightarrow\frac{\sqrt2}{3}>0.
 \label{eq:product-positive-density}
\end{equation}
Choose one fixed sufficiently small $\gamma_0<0$.  The linear term in
Eq.~\eqref{eq:product-linear-response} is then positive and dominates the
cubic remainder uniformly, giving
$\mathcal F_{0,r}(\gamma_0,\pi/8)\ge c_0>0$ for all sufficiently large $r$.
Theorem~\ref{thm:gaussian-dephasing} makes the typical entangled-family
supremum on $\beta=b/r$ smaller than $c_0/2$ for large $r$, completing the
proof of Theorem~\ref{thm:product-separation}.

No boundary-layer limit is asserted here.  Such a limit would require
finite-order joint concentration of multiple syndrome overlaps, which is not
proved by the one-branch bounded-difference argument above.

\end{document}